\documentclass[final]{amsart}

\usepackage[T1]{fontenc}
\usepackage[utf8]{inputenc}
\usepackage{amsmath,amssymb}
\usepackage{mathrsfs}
\usepackage{amsfonts,amsthm,latexsym}
\usepackage[arrow, matrix, curve]{xy}
\usepackage{mathtools}
\usepackage{nicefrac}
\usepackage{tikz}
\usepackage{bracket-hack}
\usepackage{ming}
\usepackage[english]{babel}
\usepackage{graphicx}
\usepackage{color}
\usepackage{paralist}
\usepackage[final]{hyperref}
\usepackage[colorinlistoftodos,obeyDraft]{todonotes}
\usepackage[notcite,notref]{showkeys}
\usepackage{marginnote}
\usepackage{ifdraft}
\usepackage[a4paper,twoside]{geometry}

\mathtoolsset{showonlyrefs=true}
\mathtoolsset{centercolon}
\setdefaultenum{i)}{}{}{}
\usepackage{auto-eqlabel}

\theoremstyle{definition}
\newtheorem{defi}{Definition}[section]

\newtheorem{rema}[defi]{Remark}

\newtheorem{inducassum}[defi]{Inductive assumption}

\theoremstyle{plain}
\newtheorem{prop}[defi]{Proposition}
\newtheorem{theo}[defi]{Theorem}
\newtheorem{conj}[defi]{Conjecture}
\newtheorem{lemm}[defi]{Lemma}
\newtheorem{coro}[defi]{Corollary}

\newenvironment{subaligned}{\left\{\aligned}{\endaligned\right.}

\newcommand{\A}{{\mathcal A}}

\newcommand{\B}{{\mathcal B}}

\newcommand{\D}{{\mathcal D}}

\renewcommand{\H}{{\mathcal H}}
\renewcommand{\O}{{\mathcal O}}

\newcommand{\maths}[1]{{\ensuremath{\mathbb #1}}}
\newcommand{\RR}{\maths{R}}
\newcommand{\NN}{\maths{N}}

\renewcommand{\SS}{\maths{S}}

\newcommand{\ZZ}{\maths{Z}}

\newcommand{\TT}{\maths{T}}

\newcommand{\ie}{i.\,e.\ }

\newcommand{\rhs}{right-hand side}
\newcommand{\lhs}{left-hand side}

\newcommand{\mf}{mani\-fold}

\newcommand{\fundform}{second fun\-da\-men\-tal form}

\newcommand{\onorm}{or\-tho\-nor\-mal}

\newcommand{\levi}{Levi-Civita}

\newcommand{\desitter}{de~Sitter}

\newcommand{\eps}{\varepsilon}

\newcommand{\absval}[1]{\lvert #1 \rvert}
\newcommand{\norm}[1]{\lVert #1 \rVert}
\newcommand{\normp}[2]{\lVert #2 \rVert_{#1}}

\newcommand{\normck}[2]{\lVert #2 \rVert _{\ck{#1}}}
\newcommand{\normCzero}[1]{\normp{C^0}{#1}}

\newcommand{\ck}[1]{C^{#1}}

\newcommand{\mean}[1]{\langle #1 \rangle}

\begin{document}

\title[On Cosmic No-Hair in~\texorpdfstring{$\TT^2$}{T2}-symmetric non-linear scalar fields]{On the Cosmic No-Hair Conjecture in~\texorpdfstring{$\TT^2$}{T2}-symmetric non-linear scalar field spacetimes}
\author[K. Radermacher]{Katharina Radermacher}
\address{Department of Mathematics, KTH Royal Institute of Technology, SE-10044 Stockholm, Sweden}
\curraddr{}
\email{kmra@kth.se}
\urladdr{}
\dedicatory{}
\date{\today}
\translator{}

\newcommand{\Ebas}{E_{\operatorname{bas}}}
\newcommand{\Enew}{E_{\operatorname{exist}}}
\newcommand{\potV}{V}
\newcommand{\constpot}{\potV_{\operatorname{const}}}
\newcommand{\maxpot}{\potV_{\operatorname{max}}}
\newcommand{\Vasympt}{\potV_{\operatorname{asympt}}}
\newcommand{\scalphi}{\phi}
\newcommand{\partt}{\partial_t}
\newcommand{\parttheta}{\partial_\theta}
\newcommand{\partx}{\partial_x}
\newcommand{\party}{\partial_y}
\newcommand{\partpm}{\partial_\pm}
\newcommand{\partmp}{\partial_\mp}
\newcommand{\Afct}{\A}
\newcommand{\Afctpm}{\Afct_\pm}
\newcommand{\Afctmp}{\Afct_\mp}
\newcommand{\hatAfctpm}{\hat\Afct_\pm}
\newcommand{\hatAfctmp}{\hat\Afct_\mp}
\newcommand{\Afcttwopm}{\Afct_{\pm,2}}
\newcommand{\Afcttwomp}{\Afct_{\mp,2}}
\newcommand{\Afcttwoplus}{\Afct_{+,2}}
\newcommand{\Afcttwominus}{\Afct_{-,2}}
\newcommand{\Bfct}{\B}
\newcommand{\Bfctpm}{\Bfct_\pm}
\newcommand{\Bfctmp}{\Bfct_\mp}
\newcommand{\Bfcttwopm}{\Bfct_{\pm,2}}
\newcommand{\Bfcttwomp}{\Bfct_{\mp,2}}
\newcommand{\Bfcttwoplus}{\Bfct_{+,2}}
\newcommand{\Bfcttwominus}{\Bfct_{-,2}}
\newcommand{\hatBfctpm}{\hat\Bfct_\pm}
\newcommand{\hatBfctmp}{\hat\Bfct_\mp}
\newcommand{\DfctNp}{\D_{N+1,+}}
\newcommand{\DfctNm}{\D_{N+1,-}}
\newcommand{\DfctNpm}{\D_{N+1,\pm}}
\newcommand{\DfctNmp}{\D_{N+1,\mp}}
\newcommand{\hatDfctNp}{\hat\D_{N+1,+}}
\newcommand{\hatDfctNm}{\hat\D_{N+1,-}}
\newcommand{\hatDfctNpm}{\hat\D_{N+1,\pm}}
\newcommand{\hatDfctNmp}{\hat\D_{N+1,\mp}}
\newcommand{\quotJ}[1]{\frac{e^{\lambda/2-P}J^2}{t^{{#1}/2}}}
\newcommand{\quotKQJ}[1]{\frac{e^{\lambda/2+P}(K-QJ)^2}{t^{{#1}/2}}}
\newcommand{\boundpot}{C_{pot}}
\newcommand{\nohairconj}{Cosmic No-Hair conjecture}
\newcommand{\setHKT}{C_{\H,K,T}}
\newcommand{\nabladesitter}{\overline\nabla_{\operatorname{dS}}}

\newcommand{\consevolequsinterval}{Consider a solution to the evolution equations~\eqref{eqn_timeevolalpha}--\eqref{eqn_eqnofmotionscalfieldwithT2} on~$I\times\SS^1$. }
\newcommand{\consevolequscompinterval}{Consider a solution to the evolution equations~\eqref{eqn_timeevolalpha}--\eqref{eqn_eqnofmotionscalfieldwithT2} on~$I\times\SS^1$. Assume that the interval~$I=[t_0,t]$ is compact. }
\newcommand{\conslocalsolTtwo}{Consider a solution to equations~\eqref{eqn_timeevolalpha}--\eqref{eqn_eqnofmotionscalfieldwithT2} on~$I\times\SS^1$. }
\newcommand{\consconstVandTtwoandlambdaasympt}{Consider the Einstein non-linear scalar field equations with a positive constant potential~$\constpot>0$ and a future global~$\TT^2$-symmetric solution to these equations which has~$\lambda$-asymptotics with constant~$\constpot$. }
\newcommand{\dependconstant}{depending only on the initial data at time~$t_0$ and polynomially on the length of the interval}
\newcommand{\smoothfctdependconstant}{depending only on the initial data at time~$t_0$ and via a smooth~$\RR\rightarrow\RR$ function on the length of the interval}
\newcommand{\assumptnonnegatpotential}{Assume further that the potential~$\potV:\RR\rightarrow\RR$ is a non-negative~$C^\infty$ function. }

\begin{abstract}	
	We consider spacetimes solving the Einstein non-linear scalar field equations with~$\TT^2$-symmetry and show that they admit an areal time foliation in the expanding direction. In particular, we prove global existence and uniqueness of solutions to the corresponding system of evolution equations for all future times. The only assumption we have to make is that the potential is a non-negative smooth function.
	
	In the special case of a constant potential, a setting which is equivalent to a linear scalar field on a background with a positive cosmological constant, we achieve detailed asymptotic estimates for the different components of the spacetime metric. This result holds for all $\TT^3$-Gowdy symmetric metrics and extends to certain~$\TT^2$-symmetric ones satisfying an a priori decay property. 
	Building upon these asymptotic estimates, we show future causal geodesic completeness and prove the \nohairconj.
\end{abstract}

\maketitle

\section{Introduction}

When considering Einstein's field equations in the cosmological setting of General Relativity, the current preference is to choose a non-linear scalar field or a positive cosmological constant as a part of the stress-energy tensor. The reason for these choices is that they are consistent with supernova observations which strongly hint at an accelerated expansion of the universe, see~\cite{riessetal_observevidsupernovaeaccelunivcosmconst}. 
Non-linear scalar field models in particular are used for models of the universe at very early times, when modelling inflatons, as well as for models of the late time expansion referred to as quintessence.
In the present paper, we consider spacetimes solving the Einstein non-linear scalar field equations. The special case of a constant potential coincides with a linear scalar field on a background with a positive cosmological constant.

We focus on spacetimes with~$\TT^2$-symmetry, which constitute an intermediate step between the well-studied spatially homogeneous models and the fully general case without any imposed symmetry. In this setting, we investigate the late time behaviour of solutions to the Einstein non-linear scalar field equations. We show future global existence of solutions under mild assumptions on the potential. In the case of a constant potential, we find future causal geodesic completeness and prove the \nohairconj\ in case of~$\TT^3$-Gowdy symmetry and certain~$\TT^2$-symmetric solutions.

\subsection{The \nohairconj}

The main question when studying the late time behaviour of cosmological solutions to Einstein's field equations, at least when assuming a positive cosmological constant, is the \nohairconj. This conjecture states that, in general, solutions are expected to isotropise in the expanding direction, in the sense that the geometry as perceived by observers becomes \desitter\ like. 
We give a precise formulation of the conjecture in our setting in Conjecture~\ref{conj_nohair}.

\subsection{Matter assumptions: Non-linear scalar field}

When modelling accelerated expansion in spacetimes, the easiest way to achieve this is by incorporating a positive cosmological constant in the Einstein field equations. A more sophisticated version to obtain the same behaviour is to consider a non-linear scalar field. In this matter model, a \textit{potential~$\potV:\RR\rightarrow\RR$} is given satisfying suitable minimal assumptions. Those usually include smoothness properties and one often assumes that the potential is non-negative or even strictly positive.

The stress-energy tensor of a non-linear scalar field is of the form
\begin{equation}
\label{eqn_stressenergytensorgeneralscalarfield}
	T_{\alpha\beta}=\nabla_\alpha\scalphi\nabla_\beta\scalphi-\left[\frac12\nabla_\gamma\scalphi\nabla^\gamma\scalphi+\potV(\scalphi)\right]g_{\alpha\beta},
\end{equation}
where $\scalphi$ is called the \textit{non-linear scalar field}.
Further, the potential and the scalar field satisfy the relation 
\begin{equation}
\label{eqn_eqnofmotiongeneralscalarfield}
	\nabla^\alpha\nabla_\alpha\scalphi-\potV'(\scalphi)=0,
\end{equation}
which is referred to as \textit{equation of motion}.

The Einstein non-linear scalar field equations then read
\begin{equation}
	\operatorname{Ric}-\frac12 Sg=\operatorname{Ein}=T,
\end{equation}
where~$\operatorname{Ric}$ and~$S$ are the Ricci and scalar curvature of the spacetime whose metric is~$g$.

One notices that the assumption of a constant positive potential~$\potV=\potV_0>0$ is equivalent to a linear scalar field on a background with a positive cosmological constant. We discuss this case in detail in Section~\ref{section_constpotV}, proving that the \nohairconj\ holds under certain additional assumptions on the symmetry.

\subsection{Symmetry assumptions: \texorpdfstring{$\TT^2$}{T2} and Gowdy}

In this paper, we study~$\TT^2$-symmetric spacetimes, and one way of doing so is by imposing the following setup:
One assumes that the topology is~$I\times\TT^3$, where~$\TT^3$ is the three-dimensional torus and~$I\subset(0,\infty)$ an interval. The metric of the spacetime is assumed to have the form
\begin{equation}
\label{eqn_metricT2}
	g=t^{-1/2}e^{\lambda/2}({-}dt^2+\alpha^{-1}d\theta^2)+te^P[dx+Qdy+(G+QH)d\theta]^2+te^{{-}P}(dy+Hd\theta)^2,
\end{equation}
where $(\theta,x,y)$ are the standard coordinates of the torus and~$t$ the one of the interval. The functions~$\alpha>0$, $\lambda$, $P$, $Q$, $G$ and~$H$ are assumed to depend only on~$t$ and~$\theta$.
Due to these functions being invariant under changes in the~$x,y$ coordinates, the torus~$\TT^2$ acts smoothly on the spacetime via translation in~$x$ and~$y$, leaving the metric invariant. Furthermore, one notices that the area of the symmetry orbits~$\{t\}\times\{\theta\}\times\TT^2$ is proportional to~$t$. 
One therefore denotes the foliation corresponding to the metric~\eqref{eqn_metricT2} by~\textit{areal time foliation}.

The spacetime is called~\textit{$\TT^3$-Gowdy symmetric} if the twist quantities defined by
\begin{equation}
\label{eqn_twistgeometric}
	J=\epsilon_{\alpha\beta\gamma\delta}X^\alpha Y^\beta \nabla^\gamma X^\delta,\qquad
	K=\epsilon_{\alpha\beta\gamma\delta}X^\alpha Y^\beta \nabla^\gamma Y^\delta
\end{equation}
vanish, where~$X=\partial_x$ and~$Y=\partial_y$ are  the Killing fields of the metric~\eqref{eqn_metricT2} and~$\epsilon$ is its volume form. In terms of the variables appearing in the metric, $\TT^3$-Gowdy symmetry corresponds to~$G$ and~$H$ being time-independent.

\subsection{Previous results}

Models with~$\TT^2$-symmetry have experienced a lot of attention, mainly in the last two decades, as they are considered a stepping stone towards more general geometric assumptions. One important objective has been to prove existence and properties of an areal time foliation, an approach which was started in~\cite{bergerchruscielisenbergmoncrief_globalfoliatvacuumT2} with a discussion of the vacuum case. In the same setting, \cite{isenbergweaver_areasymmetryorbitsT2} shows that the areal coordinates extend towards the past up to time parameter~$R=0$.
In~\cite{clausenisenberg_arealfolAVTDinT2poscosmconst}, the vacuum $\TT^2$-symmetric setting is generalised to even include a positive cosmological constant. Among other statements, it is shown that also in this setting, there exists an areal time foliation for the whole spacetime.

In a matter model different from the one discussed here, namely for solutions to the Einstein-Vlasov equations, similar results have been achieved:
Global existence of an areal foliation is shown to hold for $\TT^3$-Gowdy symmetry in~\cite{andreasson_globalfoliatgowdysymm}. In~\cite{andreassonrendallweaver_existCMCT2Vlasov}, this statement is extended to~$\TT^2$-symmetric solutions.
The above papers culminated in~\cite{smulevici_areasymmetrytoroidhyperb}, which collects and proves global existence theorems and areal foliations for a number of geometric settings with toroidal and hyperbolic symmetry ($\TT^2$-symmetry included).
The matter models discussed are vacuum and Vlasov matter, optionally including a positive cosmological constant.
We refer in particular to Table~1 and~2 in that reference, giving details on the minimal possible value of the time parameter and details on the status of the Strong Cosmic Censorship conjecture in such spacetimes.
\cite{tegankong_einsteinvlasovscalarfieldgowdyexpanddirection} focusses on $\TT^3$-Gowdy symmetry and the Einstein-Vlasov equations with a linear scalar field and shows that solutions are future global.

Regarding the \nohairconj, only few results have been obtained in settings directly related to ours: In \cite{wald_asymptbehavhomogmodelspositcosmconst}, the late time behaviour and \nohairconj\ in spatially homogeneous spacetimes (Bianchi) is discussed for solutions to Einstein's equations with a positive cosmological constant and any matter type satisfying the dominant energy condition. 
The paper~\cite{andreassonringstrom_cosmicnohairT3gowdyeinsteinvlasov} treats in detail solutions to the Einstein-Vlasov equations with a positive cosmological constant. The \nohairconj\ is shown to hold in~$\TT^3$-Gowdy symmetry and extends to certain~$\TT^2$-symmetric solutions, namely those with~$\lambda$-asymptotics.

Turning to the matter model discussed in the present paper, there is extensive literature covering solutions to Einstein's equations with a positive cosmological constant. On the other hand, solutions to the Einstein non-linear scalar field equations have been discussed, though often under stronger symmetry assumptions than ours, such as homogeneity or even isotropy. 
For example, \cite{rendall_accercosmexpscalfieldpositlowerbound} investigates the dynamics of spatially homogeneous non-linear scalar field solutions, mostly of Bianchi type~I--VIII, where the potential has a positive lower bound.
\cite{hayoung_einsteinvlasovscalarfield} adds a proof of global existence in the homogeneous setting, and investigates in more detail the asymptotics in case of an exponential potential.

Inhomogeneous spacetimes solving the Einstein non-linear scalar field equations however have not received much attention so far. 
In~\cite{heinzlerendall_powerlawinflation}, the potential of the non-linear scalar field is assumed to be exponential, and the late time asymptotics are discussed without making assumptions on the symmetry.
\cite{ringstrom_futstabnonlinscalfield} discusses future global non-linear stability in case the potential has a positive minimum, making only local assumptions on the initial data and no assumption on the symmetry. In~\cite{ringstrom_powerlawinflation}, this result is extended to exponential potentials, for initial data which in a certain ball is close to locally spatially homogeneous initial data.
A more comprehensive discussion of (among other things) non-linear scalar fields, touching upon a number of topics related to this matter model, is given in~\cite{ringstrom_topologyfuturestabilityuniverse}.

\subsection{The evolution equations}

The setting which we study in this paper is that of a non-linear scalar field with~$\TT^2$-symmetry. We here give the evolution equations governing spacetime solutions in this setting. A more detailed deduction of the following equations, relying on similar equations found in~\cite{andreassonringstrom_cosmicnohairT3gowdyeinsteinvlasov} for the case of Vlasov matter with a cosmological constant, is given in Appendix~\ref{appendix_evolequsscalarfield}.

One conveniently introduces the notation 
\begin{equation}
\label{eqn_twistanalytic}
	J={-}t^{5/2}\alpha^{1/2}e^{{-}\lambda/2+P}(G_t+QH_t),\quad K=QJ-t^{5/2}\alpha^{1/2}e^{{-}\lambda/2-P}H_t,
\end{equation}
and notices that this coincides with the definition of the twist quantities from~\eqref{eqn_twistgeometric}.
These two functions are used to replace the functions~$G,H$ from the metric~\eqref{eqn_metricT2}. It turns out that for non-linear scalar fields, both~$J$ and~$K$ are constants.
The case of Gowdy symmetry is equivalent to~$J=0=K$.

The remaining quantities~$\alpha>0$, $\lambda$, $P$, $Q$ and~$\scalphi$ evolve according to the following equations, which are a direct consequence of the Einstein field equations:
\begin{align}
\label{eqn_timeevolalpha}
	\frac{\alpha_t}{\alpha}
		={}&{-}\frac{e^{\lambda/2-P}J^2}{t^{5/2}} - \frac{e^{\lambda/2+P}(K-QJ)^2}{t^{5/2}} - 4t^{1/2}e^{\lambda/2}\potV(\scalphi),\\
\label{eqn_timeevollambda}
	\lambda_t
		={}&t\left[P_t^2+\alpha P_\theta^2+e^{2P}(Q_t^2+\alpha Q_\theta^2)+2(\scalphi_t^2+\alpha\scalphi_\theta^2)\right]-\frac{e^{\lambda/2-P}J^2}{t^{5/2}}-\frac{e^{\lambda/2+P}(K-QJ)^2}{t^{5/2}}\\
			&{}-4t^{1/2}e^{\lambda/2}\potV(\scalphi),\\
\label{eqn_spaceevollambda}
	\lambda_\theta
		={}&2t(P_tP_\theta +e^{2P}Q_tQ_\theta+2\scalphi_t\scalphi_\theta),
\end{align}
further, second order differential equations are imposed on~$P$,
\begin{align}
\label{eqn_secondevolP}
	\partial_t(t\alpha^{{-}1/2}P_t)={}&\parttheta(t\alpha^{1/2}P_\theta)+t\alpha^{{-}1/2}e^{2P}(Q_t^2-\alpha Q_\theta^2)\\
	&{}+\frac{\alpha^{{-}1/2}e^{\lambda/2-P}J^2}{2t^{5/2}}-\frac{\alpha^{{-}1/2}e^{\lambda/2+P}(K-QJ)^2}{2t^{5/2}},
\end{align}
on~$Q$,
\begin{equation}
\label{eqn_secondevolQ}
	\partial_t(t\alpha^{{-}1/2}e^{2P}Q_t)-\partial_\theta(t\alpha^{1/2}e^{2P}Q_\theta)
	  =t^{{-}5/2}\alpha^{{-}1/2}e^{\lambda/2+P}J(K-QJ),
\end{equation}
and on a term containing~$\alpha$ and~$\lambda$,
\begin{align}
\label{eqn_secondevollambda}
	\partial_t\left[t\alpha^{{-}1/2}(\lambda_t-2\frac{\alpha_t}{\alpha}-\frac3t)\right]
		={}&\partial_\theta(t\alpha^{1/2}\lambda_\theta) +\alpha^{{-}1/2}\lambda_t \\
		  &{} -t\alpha^{{-}1/2}\left[P_t^2+e^{2P}Q_t^2+2\scalphi_t^2-\alpha(P_\theta^2+e^{2P}Q_\theta^2+2\scalphi_\theta^2)\right]\\
		  &{}-2\alpha^{{-}1/2}(\frac{e^{\lambda/2-P}J^2}{t^{5/2}}+\frac{e^{\lambda/2+P}(K-QJ)^2}{t^{5/2}}) +4t^{1/2}e^{\lambda/2}\alpha^{{-}1/2}\potV.
\end{align}
The potential~$\potV$ is related to the scalar field~$\scalphi$ via the equation of motion
\begin{equation}
	\label{eqn_eqnofmotionscalfieldwithT2}
	\potV'=t^{1/2}e^{{-}\lambda/2}\left[-\scalphi_{tt}+\alpha\scalphi_{\theta\theta}-\frac1t\scalphi_t+\frac{\alpha_t}{2\alpha}\scalphi_t+\frac{\alpha_\theta}{2}\scalphi_\theta \right].
\end{equation}

We notice that if we assume the evolution equations~\eqref{eqn_timeevolalpha}--\eqref{eqn_secondevolQ} to hold, then equation~\eqref{eqn_secondevollambda} holds as well. Further, if we set~$h$ to be the \lhs\ minus the \rhs\ of equation~\eqref{eqn_spaceevollambda}, then equations~\eqref{eqn_timeevolalpha}, \eqref{eqn_timeevollambda} together with~\eqref{eqn_secondevolP}, \eqref{eqn_secondevolQ} imply that
\begin{equation}
\label{eqn_evoleqnisconstraint}
	h_t=\frac{\alpha_t}{2\alpha} h.
\end{equation}
We can therefore consider equation~\eqref{eqn_spaceevollambda} to be a constraint, as it is satisfied on the whole interval of existence as soon as it is satisfied at the initial time.
In both arguments, we make use of the fact that~$J$ and~$K$ are constants. The equivalent statement in the vacuum setting was achieved in~\cite{ringstrom_instabspathomogT2vacuum}.
\begin{rema}
	Direct computation of the individual terms shows that the equation of motion~\eqref{eqn_eqnofmotionscalfieldwithT2} can be written in the form
	\begin{equation}
	\label{eqn_secondevolscalphifromeqnmotion}
		\partial_t(t\alpha^{{-}1/2}\scalphi_t)-\parttheta(t\alpha^{1/2}\scalphi_\theta) ={-}t^{1/2}\alpha^{{-}1/2}e^{\lambda/2}\potV'(\scalphi),
	\end{equation}
	which strongly ressembles the second order equations for~$P$ and~$Q$, equations~\eqref{eqn_secondevolP} and~\eqref{eqn_secondevolQ}.
	These two latter equations also appear in a different form further down, see equations~\eqref{eqn_secondevolPrewritten} and~\eqref{eqn_secondevolQrewritten}.
\end{rema}

\subsection{New results}

A priori, it is not clear that the set of evolution equations~\eqref{eqn_timeevolalpha}--\eqref{eqn_eqnofmotionscalfieldwithT2} has a unique solution nor whether solutions exist more than locally. In order to prove global existence and uniqueness towards the future, which is the time direction we are interested in, we first show that all variables are uniformly bounded on compact time intervals~$I=[t_0,t]$, by constants \smoothfctdependconstant, and that this even holds for their derivatives. This is done in Section~\ref{section_unifboundcomptime}.
In order to control the potential~$\potV$, we impose mild conditions on its behaviour as a function of~$\scalphi$, namely that it is a non-negative~$C^\infty$ function on~$\RR$.

Making use of these boundedness results, we are in a position to connect our set of evolution equations to a general statement on existence and uniqueness of hyperbolic partial differential equations from~\cite{majda_comprfluidsconservlaws}. This yields the following theorem, whose proof is carried out in Section~\ref{section_globalexist}.
\begin{theo}[Future global existence and uniqueness]
\label{theo_globalexistence}
	Consider the evolution equations~\eqref{eqn_timeevolalpha}--\eqref{eqn_eqnofmotionscalfieldwithT2}
	where the potential~$\potV:\RR\rightarrow\RR$ is a non-negative~$C^\infty$ function. 
	To given smooth initial data at time~$t_0$, there exists a unique smooth solution to these evolution equations, and it is global towards the future, \ie it is defined on~$[t_0,\infty)\times\TT^3$.
\end{theo}

Once future global existence is ensured, we can turn our attention to the main aim of the present paper, namely to understand in detail the asymptotic behaviour of the individual variables, and thus of the metric in~\eqref{eqn_metricT2}, towards the future. 

The following subclass of metrics is of central importance in our discussion:
\begin{defi}
\label{defi_lambdaasympt}
	Consider a metric of the form~\eqref{eqn_metricT2} which is defined for all~$t>t_0$ for some~$t_0\ge0$. The metric is said to have~\textit{$\lambda$-asymptotics with constant~$\Vasympt>0$} if, for every~$\eps>0$, there is a~$T>t_0$ such that
	\begin{equation}
		\lambda(t,\theta)\ge{-}3\ln t +2\ln(\frac{3}{4\Vasympt})-\eps,
	\end{equation}
	for all~$(t,\theta)\in[T,\infty)\times\SS^1$.
\end{defi}
We show below, in Lemma~\ref{lemm_lambdaasympt}, that~$\TT^3$-Gowdy symmetric solutions have~$\lambda$-asymptotics with constant~$\maxpot$,
provided we assume that the potential is bounded from above towards the future, \ie
\begin{equation}
	\potV\left[\scalphi(t,\theta)\right]\le\maxpot,
\end{equation}
for all~$(t,\theta)\in [t_0,\infty)\times\SS^1$.
For non-Gowdy solutions, we impose this property on~$\lambda$ for our discussion of the asymptotic behaviour. This approach should be compared to that of~\cite{andreassonringstrom_cosmicnohairT3gowdyeinsteinvlasov}, where the case of Vlasov matter with a positive cosmological constant is discussed. In that paper, a related definition of~$\lambda$-asymptotics is given, which is shown to hold in case of~$\TT^3$-Gowdy symmetry and imposed otherwise, see~\cite[Def.~1]{andreassonringstrom_cosmicnohairT3gowdyeinsteinvlasov}.

\medskip

We notice that the potential~$\potV$ appears in most of the evolution equations and therefore influences all variables' behaviour.
It therefore appears impossible to achieve statements on the asymptotic convergence and decay properties of the variables without imposing some a priori assumption on the potential~$\potV$. In the present paper, in order to obtain detailed asymptotics, we assume that the potential~$\potV$ is constant, $\potV=\constpot$. The resulting field equations are equivalent to those of a linear scalar field on a background with a positive cosmological constant. In this setup and with the assumption of~$\lambda$-asymptotics with the same constant~$\constpot$, we can determine the behaviour of all variables up to arbitrary high order. This is the statement of Proposition~\ref{prop_mainasympmetric}.

Although the restriction to a constant potential is a rather severe one, we hope that our findings can, in a future paper, be extended to include a larger class of potentials. 

The steps taken to achieve Proposition~\ref{prop_mainasympmetric} are laid out in Section~\ref{section_constpotV}. The starting point is the property of~$\lambda$-asymptotics with constant~$\constpot$, from which we obtain estimates for all non-differentiated variables. These are used in the following steps to control first order derivatives. Inductively, we then extend the statements to include derivatives of all orders.
The steps of our discussion are similar to those taken in~\cite{andreassonringstrom_cosmicnohairT3gowdyeinsteinvlasov}, where the case of Vlasov matter with a positive cosmological constant is treated. 
There are two main differences though: While many statements in~\cite{andreassonringstrom_cosmicnohairT3gowdyeinsteinvlasov} make use of the non-negative pressure condition, we do not have access to this property, as non-linear scalar fields in general do not satisfy this condition. Further, Vlasov matter does not include a potential~$\potV$ and a scalar field~$\scalphi$. All statements on these quantities are new for this reason.

\begin{prop}[Asymptotic properties of the variables in the metric]
\label{prop_mainasympmetric}
	Consider a~$\TT^2$-sym\-metric solution to the Einstein non-linear scalar field equations with a positive constant potential~$\constpot$. Choose coordinates so that the corresponding metric takes the form~\eqref{eqn_metricT2} on~$(t_0,\infty)\times \TT^3$. Assume that the metric has~$\lambda$-asymptotics with constant~$\constpot$ and set~$t_1=t_0+2$. 
	Then there are smooth functions~$\alpha_\infty>0$, $P_\infty$, $Q_\infty$, $\scalphi_\infty$, $G_\infty$ and~$H_\infty$ on~$\SS^1$, and, for every~$0\le N\in\ZZ$, a constant~$C_N$ such that
	\begin{align}
		\normck N{P_t(t,\cdot)} + \normck N{Q_t(t,\cdot)} + \normck N{\scalphi_t(t,\cdot)} \le{}& C_N t^{{-}2},\\
		\normck N{P(t,\cdot)-P_\infty} + \normck N{Q(t,\cdot)-Q_\infty} + \normck N{\scalphi(t,\cdot)-\scalphi_\infty} \le{}& C_N t^{{-}1},\\
		\normck N{\frac{\alpha_t}{\alpha}+\frac3t} + \normck N{\lambda_t+\frac3t} \le{}& C_k t^{{-}2},\\
		\normck N{t^3\alpha(t,\cdot)-\alpha_\infty} +\normck N{\lambda(t,\cdot)+3\ln t-2\ln\frac{3}{4\constpot}} \le{}& C_N t^{{-}1},\\
		t\normck N{G_t(,\cdot)} +t\normck N{H_t(,\cdot)} +\normck N{G(t,\cdot)-G_\infty} +\normck N{H(t,\cdot)-H_\infty}\le{}& C_Nt^{{-}3/2},
	\end{align}
	for all~$t\ge t_1$. 
	
	In terms of the geometry, denote by~$\bar g(t,\cdot)$ and~$\bar k(t,\cdot)$ the metric and \fundform\ induced on~$\{t\}\times\TT^3$, by~$\bar g_{ij}$ the components of the metric~$\bar g$ with respect to the vector fields~$\partial_1=\partial_\theta$, $\partial_2=\partial_x$ and~$\partial_3=\partial_y$, and equivalently for~$\bar k$.
	Then
	\begin{equation}
		\normck N{t^{{-}1}\bar g_{ij}(t,\cdot) - \bar g_{\infty,ij}}+ \normck N{t^{{-}1}\bar k_{ij} -\H \bar g_{\infty,ij}}\le C_N t^{{-}1},
	\end{equation}
	for all~$t\ge t_1$, where~$\H = (\constpot/3)^{1/2}$ and 
	\begin{equation}
		\bar g_\infty=\frac{3}{4\constpot\alpha_\infty}d\theta^2 + e^{2P_\infty}\left[dx + Q_\infty dy + (G_\infty+Q_\infty H_\infty)d\theta\right]^2 +e^{{-}2P_\infty}(dy+H_\infty d\theta)^2.
	\end{equation}
\end{prop}
The estimates we obtain in Proposition~\ref{prop_mainasympmetric} can be used to prove future causal geodesic completeness. The proof is given further down in Section~\ref{section_mainproofs} and proceeds similarly to that of~\cite[Prop.~4]{ringstrom_futstabnonlinscalfield}, which treats the case of the Einstein non-linear scalar field equations with a potential satisfying~$\potV(0)=\potV_0>0$, $\potV'(0)=0$, $\potV''(0)>0$.
\begin{prop}[Future causal geodesic completeness]
\label{prop_geodcomplete}
	Consider a~$\TT^2$-sym\-metric solution to the Einstein non-linear scalar field equations with a positive constant potential~$\constpot$. Choose coordinates so that the corresponding metric takes the form~\eqref{eqn_metricT2} on~$(t_0,\infty)\times \TT^3$. Assume that the metric has~$\lambda$-asymptotics with constant~$\constpot$.
	Then this spacetime is future causally geodesically complete.
\end{prop}

\medskip

Finally, we discuss how the spacetimes corresponding to the solutions to equations~\eqref{eqn_timeevolalpha}--\eqref{eqn_eqnofmotionscalfieldwithT2} appear to a late time observer. For a proof of the \nohairconj, we need to show that to such an observer the spacetime appears as the \desitter\ spacetime. Let us therefore consider the metric
\begin{equation}
\label{eqn_desitter}
	g_{\operatorname{dS}}={-}dt^2+e^{2\H t}\bar g_{\operatorname{Eucl}},
\end{equation}
where~$(\RR^3,\bar g_{\operatorname{Eucl}})$ is the Euclidean three-dimensional space. Then~$(\RR^4,g_{\operatorname{dS}})$ is a part of the \desitter\ spacetime. An observer moves along a future directed causal curve~$\gamma=(\gamma^0,\bar\gamma)$, which we assume to be future inextendible and defined on~$(s_-,s_+)$. For~$s\nearrow s_+$, the spatial part~$\bar\gamma(s)$ converges to some~$\bar x_0\in\RR^3$.
The form of the metric reveals that the curve~$\gamma$ is contained in the set
\begin{equation}
	C_{\bar x_0,\H} = \{(t,\bar x) : \absval{\bar x - \bar x_0}\le\H^{{-}1}e^{{-}\H t}\}
\end{equation}
at all times. 
In other words, the spacetime outside of this cone-like set is unreachable to the late time observer, and therefore should be irrelevant for the discussion of the \nohairconj.

As we are interested in the late time behaviour, we replace the set~$C_{\bar x_0,\H}$ by one where we bound the time coordinate from below. Further, we introduce a margin in the spatial extension. This makes it possible to work with open sets, and we introduce
\begin{equation}
	\setHKT = \{(t,\bar x) :  t>T,\: \absval{\bar x} < K\H^{{-}1}e^{{-}\H t}\}
\end{equation}
as our main object of interest, where we assume that~$T>0$ and~$K\ge 1$. Note that translation on each hypersurface~$\{t\}\times\RR^3$ is an isometry by which we can move~$\bar x_0$ to the origin in~$\RR^3$.

When considering the Einstein non-linear scalar field equations and assuming the potential to be constant, the resulting equations are equivalent to a linear field on a background with a positive cosmological constant. This is one of the settings in which \cite[Def.~8]{andreassonringstrom_cosmicnohairT3gowdyeinsteinvlasov} applies and defines the concept of future asymptotically \desitter\ like spacetimes, of which we make use here.
\begin{defi}
\label{defi_asymptdeSitterlike}
	Let~$(M,g)$ be a time-oriented, globally hyperbolic Lorentz \mf\ which is future causally geodesically complete.
	Assume that~$(M,g)$ is a solution to Einstein's equations with a positive cosmological constant~$\Lambda$ and set~$\H=(\Lambda/3)^{1/2}$.
	Then~$(M,g)$ is said to be \textit{future asymptotically \desitter\ like} if there is a Cauchy hypersurface~$\Sigma$ in~$(M,g)$ such that for every future oriented and inextendible causal curve~$\gamma$ in~$(M,g)$, the following holds:
	\begin{itemize}
		\item There is an open set~$D$ in~$(M,g)$ such that~$J^-(\gamma)\cap J^+(\Sigma)\subset D$, and~$D$ is diffeomorphic to~$\setHKT$ for a suitable choice of~$K\ge1$ and~$T>0$.
		\item Using~$\psi:\setHKT\rightarrow D$ to denote the diffeomorphism;
		letting~$R(t)=K\H^{{-}1}e^{{-}\H t}$; 
		denoting by~$\bar g_{\operatorname{dS}}(t,\cdot)$ and~$\bar k_{\operatorname{dS}}(t,\cdot)$ the metric and \fundform\ induced on~$S_t=\{t\}\times B_{R(t)}(0)$ by~$g_{\operatorname{dS}}$; 
		denoting by~$\bar g(t,\cdot)$ and~$\bar k(t,\cdot)$ the metric and \fundform\ induced on~$S_t$ by the pullback~$\psi^*g$ of the metric~$g$ by~$\psi$; and letting~$N\in\NN$, one finds
		\begin{equation}
			\lim_{t\rightarrow\infty}(
			\norm{\bar g_{\operatorname{dS}}(t,\cdot)-\bar g(t,\cdot)}_{C^N_{\operatorname{dS}}(S_t)} + 
			\norm{\bar k_{\operatorname{dS}}(t,\cdot)-\bar k(t,\cdot)}_{C^N_{\operatorname{dS}}(S_t)}
			)=0.
		\end{equation}
	\end{itemize}

\end{defi}

\begin{rema}
	In the previous definition, the sets~$J^\pm$ denote the causal future and past.
	Further, the norm we have used is defined by
	\begin{equation}
		\norm{h}_{C^N_{\operatorname{dS}}(S_t)}
		=(\sup_{S_t}\sum_{l=0}^N \bar g_{\operatorname{dS},i_1j_1}\cdots \bar g_{\operatorname{dS},i_lj_l} \bar g_{\operatorname{dS}}^{im}\bar g_{\operatorname{dS}}^{jn}
		\nabladesitter^{i_1}\cdots\nabladesitter^{i_l} h_{ij} \nabladesitter^{j_1}\cdots\nabladesitter^{j_l} h_{mn})^{1/2},
	\end{equation}
	for a covariant vector field~$h$ on~$S_t$. Here, $\nabladesitter$ denotes the \levi\ connection induced by the metric~$\bar g_{\operatorname{dS}}(t,\cdot)$ which was given in Definition~\ref{defi_asymptdeSitterlike}.
\end{rema}
\begin{rema}
	In order to extend the concept of future asymptotically \desitter\ like to solutions with a non-constant potential, Definition~\ref{defi_asymptdeSitterlike} needs to be adapted conceptually, as the sets~$\setHKT$ are no longer meaningful. The necessary adaptations might depend on additional assumptions made on the potential~$\potV$.
\end{rema}

This definition enables us to formally state the \nohairconj. Our formulation should be compared to that of~\cite[Conj.~11]{andreassonringstrom_cosmicnohairT3gowdyeinsteinvlasov}.
\begin{conj}[Cosmic No-Hair, scalar field with positive constant potential]
\label{conj_nohair}
	Let~$\A$ denote the class of initial data such that the corrresponding maximal globally hyperbolic developments are future causally geodesically complete solutions to the Einstein non-linear scalar field equations with a positive constant potential~$\constpot$.
	Then every generic element of~$\A$ has a maximal globally hyperbolic development which is future asymptotically \desitter\ like.
\end{conj}

\begin{rema}
	We do not expect the statement to hold for each element of~$\A$, but allow for a small subset to have different behaviour.
	The exact definition of this subset of exceptions as well as the notion of 'smallness' might depend on the specific class~$\A$.
	
	In the case of Einstein's vacuum equations with a positive cosmological constant, which is the special case of a constant potential and a vanishing scalar field, a famous counterexample to the conjectured behaviour are the Nariai spacetimes. These are time-oriented, globally hyperbolic, and causally geodesically complete, but without future asymptotically \desitter\ like behaviour.  
	Details on these spacetimes as well as a counterexample solving the Einstein-Maxwell equations with a positive cosmological constant are given in~\cite[p.~126f.]{ringstrom_futstabnonlinscalfield}.
\end{rema}

For~$\TT^3$-Gowdy symmetric solutions, we prove that the~\nohairconj\ holds. Our statement even extends to~$\TT^2$-symmetric spacetimes which have~$\lambda$-asymptotics, see Definition~\ref{defi_lambdaasympt}, a property which is shown to hold in~$\TT^3$-Gowdy symmetric spacetimes in Lemma~\ref{lemm_lambdaasympt}. The proof of the following theorem is carried out in Section~\ref{section_mainproofs}.
\begin{theo}[Cosmic No-Hair, scalar field with positive constant potential]
\label{theo_nohair}
	Consider a~$\TT^2$-sym\-metric solution to the Einstein non-linear scalar field equations with a positive constant potential~$\constpot$.  Choose coordinates so that the corresponding metric takes the form~\eqref{eqn_metricT2} on~$(t_0,\infty)\times \TT^3$. Assume that the metric has~$\lambda$-asymptotics with constant~$\constpot$.
	Then the solution is future asymptotically \desitter\ like.
	In other words, the \nohairconj\ holds in this class of initial data.
\end{theo}

\subsection{Notation}
\label{subsection_notationmean}

Throughout the paper, we make use of the following notation for the mean of a scalar function~$f$ on~$\SS^1$:
\begin{equation}
	\mean f=\frac1{2\pi}\int_{\SS^1}f d\theta.
\end{equation}
We even extend this definition to include the spacelike mean of functions defined on~$I\times\SS^1$, writing~$\mean f=\mean{f(t,\cdot)}$.

\subsection*{Acknowledgements}

The author wishes to thank
Hans Ringström for ongoing supervision, for suggesting the topic, as well as for many interesting discussions.

Further, the author would like to acknowledge the support of the Göran Gustafsson Foundation for
Research in Natural Sciences and Medicine. 
This research was supported by the Swedish Research Council,
Reference number 621-2012-3103.

\section{Lightcone evolution equations}

The variables appearing in the metric~\eqref{eqn_metricT2} evolve according to the evolution equations~\eqref{eqn_timeevolalpha}--\eqref{eqn_eqnofmotionscalfieldwithT2}.
In addition to the~$t$- and~$\theta$-derivatives of these variables, we are interested in their evolution along characteristics. In particular, several of our arguments revolve around the quantities
\begin{equation}
\label{eqn_defiAfctpm}
	\Afctpm=(\partpm P)^2 + e^{2P}(\partpm Q)^2
\end{equation}
and 
\begin{equation}
\label{eqn_defiBfctpm}
	\Bfctpm=(\partpm \scalphi)^2,
\end{equation}
where
\begin{equation}
\label{eqn_defilightconederiv}
	\partpm=\partt\pm\alpha^{1/2}\parttheta.
\end{equation}
Similar quantities are introduced in~\cite{andreassonringstrom_cosmicnohairT3gowdyeinsteinvlasov}, where the case of Vlasov matter is discussed. For our discussion, we make use of their lightcone derivatives. We frequently use the following abbreviated notation for the different partial derivatives:
\begin{equation}
	P_t\coloneqq\partt P,\qquad P_\theta\coloneqq\parttheta P,\qquad P_\pm\coloneqq\partpm P,
\end{equation}
and equivalently for all other variables.
\begin{lemm}
\label{lemm_Afctlightconederiv}
	Consider a solution to the evolution equations~\eqref{eqn_timeevolalpha}--\eqref{eqn_eqnofmotionscalfieldwithT2}. Then
	\begin{align}
	\partpm\Afctmp
	  ={}&{-}(\frac2t-\frac{\alpha_t}\alpha)\Afctmp \mp \frac2t\alpha^{1/2}(P_\theta P_\mp +e^{2P}Q_\theta Q_\mp )\\
			&{}+\frac{e^{{-}P+\lambda/2}J^2}{t^{7/2}} P_\mp  - \frac{e^{P+\lambda/2}(K-QJ)^2}{t^{7/2}} P_\mp   + 2\frac{e^{\lambda/2}J(K-QJ)}{t^{7/2}}e^P Q_\mp .
\end{align}
\end{lemm}
The proof is a rather lengthy computation, using the evolution equations~\eqref{eqn_secondevolP} and~\eqref{eqn_secondevolQ}.
We skip the details and refer the reader to the proof of~\cite[Lemma~38]{andreassonringstrom_cosmicnohairT3gowdyeinsteinvlasov}, or alternatively to Appendix~\ref{appendix_evolequsscalarfield}, which provides the means of directly transferring the statement of~\cite[Lemma~38]{andreassonringstrom_cosmicnohairT3gowdyeinsteinvlasov} from Vlasov matter to non-linear scalar field matter.
\begin{lemm}
\label{lemm_Bfctlightconederiv}
	Consider a solution to the evolution equations~\eqref{eqn_timeevolalpha}--\eqref{eqn_eqnofmotionscalfieldwithT2}. Then
	\begin{equation}
		\partpm\Bfctmp = {-}(\frac2t-\frac{\alpha_t}{\alpha}) \Bfctmp \mp\frac2t\alpha^{1/2}\scalphi_\theta\scalphi_\mp -2t^{{-}1/2}e^{\lambda/2}\potV'\scalphi_\mp.
	\end{equation}
\end{lemm}
The proof, a straight-forward computation using the equation of motion~\eqref{eqn_eqnofmotionscalfieldwithT2}, is left to the reader.

Additionally, we define the quantities
\begin{equation}
	\Afcttwopm\coloneqq(P_{tt}\pm\alpha^{1/2}P_{t\theta})^2+ e^{2P}(Q_{tt}\pm\alpha^{1/2}Q_{t\theta})^2
\end{equation}
as well as
\begin{equation}
	\Bfcttwopm\coloneqq(\scalphi_{tt}\pm\alpha^{1/2}\scalphi_{t\theta})^2.
\end{equation}
We use their lightcone derivatives in an integration approach along characteristic curves to obtain uniform bounds for the second order derivatives of~$P$ and~$Q$, Proposition~\ref{prop_estimsecondderivPQ}, and of~$\scalphi$, Proposition~\ref{prop_estimsecondderivscalphi}.

In the same spirit, we use the lightcone derivatives of
\begin{equation}
	\DfctNpm \coloneqq\left[\parttheta^N\scalphi_t \pm \parttheta^N(\alpha^{1/2}\scalphi_\theta )\right]^2
\end{equation}
to prove precise late time estimates for higher order derivatives of~$\scalphi$ in Lemma~\ref{lemm_inductassumfornextN}.

\section{Preliminary late time properties}
\label{section_consDECprelimestim}

Before diving into a detailed discussion of the individual variables appearing in the metric and how their interlaced asymptotic behaviour can be used to prove global existence and uniqueness on the one hand as well as the \nohairconj\ on the other hand, we state three preliminary lemmata on the late time behaviour of~$\alpha$ and~$\lambda$. These statements are used in several instances further down, adapted to the respective circumstance.

\begin{lemm}
\label{lemm_estimatealphaVnonnegative}
	\conslocalsolTtwo	
	Assume that the potential~$\potV$ is non-negative. Then~$\alpha$ is monotone decreasing on~$I$, \ie
  	\begin{equation}
		\alpha(t,\theta)\le \alpha(t_0,\theta),
	\end{equation}
	for all~$\theta\in\SS^1$ and~$t\ge t_0$, where~$t,t_0\in I$.
\end{lemm}
The statement and proof of this statement and the next are similar to those of~\cite[Prop.~42]{andreassonringstrom_cosmicnohairT3gowdyeinsteinvlasov}.
\begin{proof}
The evolution equations~\eqref{eqn_timeevolalpha} and~\eqref{eqn_timeevollambda} yield that~$\lambda_t-\alpha_t/\alpha\ge0$. Due to compactness of the sphere $\SS^1$, this implies
\begin{equation}
\label{eqn_auxilestimalphaone}
	(\alpha^{-1/2}e^{\lambda/2})(t,\theta)\ge c_0,
\end{equation}
for some constant $c_0>0$ depending on the initial data at time~$t_0\in I$, and all $(t,\theta)\in I\times\SS^1$ with $t\ge t_0$. From the evolution of $\alpha_t/\alpha$, equation~\eqref{eqn_timeevolalpha}, we conclude that
\begin{equation}
\label{eqn_auxilestimalphatwo}
	{-}\alpha_t/\alpha\ge 4t^{1/2}e^{\lambda/2}\potV(\scalphi)
\end{equation}
which is non-negative by assumption. Consequently,
\begin{equation}
	\partial_t\alpha^{-1/2}={-}\frac{\alpha_t}{2\alpha}\alpha^{-1/2}\ge 0,
\end{equation}
and integration with respect to time concludes the proof.
\end{proof}
We devote one of the following chapters to the case of a constant potential. In this case, we can strengthen the statement of the previous lemma and find polynomial decay of~$\alpha$.
\begin{lemm}
\label{lemm_estimatealphaVconstant}
	\conslocalsolTtwo	
	Assume that the potential~$\potV$ is positive and constant, $\potV=\constpot>0$. Then there is a constant $C>0$ depending only on the initial data at time~$t_0$ such that
  	\begin{equation}
		\alpha(t,\theta)\le Ct^{-3},
	\end{equation}
	for all~$\theta\in\SS^1$ and~$t\ge t_0$, where~$t,t_0\in I$.
\end{lemm}
\begin{proof}
As in the previous proof, we conclude the inequalities~\eqref{eqn_auxilestimalphaone} and~\eqref{eqn_auxilestimalphatwo}.  
By assumption, the potential~$\potV$ is a positive constant, and consequently
\begin{equation}
	\partial_t\alpha^{-1/2}={-}\frac{\alpha_t}{2\alpha}\alpha^{-1/2}\ge c_1t^{1/2},
\end{equation}
for all $(t,\theta)\in I\times\SS^1$, where the constant $c_1>0$ depends only on the initial data at time~$t_0$. Integration with respect to time concludes the proof.
\end{proof}

In the next lemma, we take a first step towards our proof of the \nohairconj. We assume for the moment that the solution to the evolution equations extends to all future times, and we assume symmetry to hold. Later, we will consider constant potentials, but for the next lemma, it is enough to assume that the potential is bounded from above towards the future.
\begin{lemm}
\label{lemm_lambdaasympt}
	Consider a solution to equations~\eqref{eqn_timeevolalpha}--\eqref{eqn_eqnofmotionscalfieldwithT2} on~$[t_0,\infty)\times\SS^1$. 
	Assume that the solution has~$\TT^3$-Gowdy symmetry, \ie $J=0=K$, and assume that the potential is bounded from above towards the future, \ie
	\begin{equation}
		\potV\left[\scalphi(t,\theta)\right]\le\maxpot,
	\end{equation}
	for all~$(t,\theta)\in [t_0,\infty)\times\SS^1$.
	Then there is, for every~$\eps>0$, a time~$T>t_0$ such that
	\begin{equation}
		\lambda(t,\theta)\ge {-} 3\ln t + 2 \ln(\frac{3}{4\maxpot})-\eps,
	\end{equation}
	for all~$(t,\theta)\in[t_0,\infty)\times\SS^1$.
\end{lemm}
The statement and proof are similar to~\cite[Prop.~44]{andreassonringstrom_cosmicnohairT3gowdyeinsteinvlasov}. Note however that their statement assumes the non-negative pressure condition to hold, which in general is not satisfied for non-linear scalar field spacetimes.
\begin{proof}
	We set
	\begin{equation}
	\label{eqn_defihatlambda}
		\hat\lambda(t,\theta)=\lambda(t,\theta) + 3\ln t - 2\ln(\frac{3}{4\maxpot}).
	\end{equation}
	The evolution equation for~$\lambda$, equation~\eqref{eqn_timeevollambda}, together with~$J=0=K$, gives
	\begin{equation}
	\label{eqn_evolhatlambda}
		\partt\hat\lambda=t\left[P_t^2+\alpha P_\theta^2+e^{2P}(Q_t^2+\alpha Q_\theta^2)+2(\scalphi_t^2+\alpha\scalphi_\theta^2)\right]
			-4t^{1/2}e^{\lambda/2}\potV(\scalphi)+\frac3t.
	\end{equation}
	Due to the assumption on the potential, we can estimate this expression by
	\begin{equation}
		\partt\hat\lambda \ge \frac3t (1- e^{\hat\lambda/2}).
	\end{equation}
	Consequently, for every~$\eps>0$, there is a~$T$ such that~$\hat\lambda(t,\theta)\ge{-}\eps$ for all~$(t,\theta)\in[T,\infty)\times\SS^1$. This concludes the proof.
\end{proof}

\section{Uniform boundedness on compact time intervals}
\label{section_unifboundcomptime}

In this section, we consider solutions to the evolution equations~\eqref{eqn_timeevolalpha}--\eqref{eqn_eqnofmotionscalfieldwithT2} given on compact time intervals~$I=[t_0,t]$. We show that towards the future all quantities are uniformly bounded up to their second derivative, by constants \smoothfctdependconstant. In fact, our chain of arguments can be extended up to derivatives of arbitrary high degree, see Remark~\ref{rema_unifboundsextendhigherderiv}, but we don't give detailed proofs for derivatives higher than second degree.

In the following section, we make use of the bounds obtained here in order to show future global existence and uniqueness of solutions. This approach is similar to existing work on~$\TT^2$-symmetric spacetimes, and we refer the reader to~\cite{andreasson_globalfoliatgowdysymm}, \cite{isenbergweaver_areasymmetryorbitsT2} and~\cite{andreassonrendallweaver_existCMCT2Vlasov} for comparison.

\begin{theo}
\label{theo_unifbounds}
	Consider the evolution equations~\eqref{eqn_timeevolalpha}--\eqref{eqn_eqnofmotionscalfieldwithT2} where the potential~$\potV:\RR\rightarrow\RR$ is assumed to be a non-negative~$C^\infty$ function. 
	Given a solution defined on~$I\times\SS^1$, where~$I=[t_0,t]$ is a compact interval, the following quantities are uniformly bounded by constants \smoothfctdependconstant:
	\begin{itemize}
		\item The zeroth derivatives $\alpha$, $\alpha^{{-}1}$, $\lambda$, $P$, $Q$, $\scalphi$, $\potV$;
		\item The first derivatives~$\alpha_t$, $\alpha_\theta$, $\lambda_t$, $\lambda_\theta$, $P_t$, $P_\theta$, $Q_t$, $Q_\theta$, $\scalphi_t$, $\scalphi_\theta$, $V'$;
		\item The second derivatives~$P_{tt}$, $P_{t\theta}$, $P_{\theta\theta}$, $Q_{tt}$, $Q_{t\theta}$, $Q_{\theta\theta}$, $\scalphi_{tt}$, $\scalphi_{t\theta}$, $\scalphi_{\theta\theta}$, $V''$.
	\end{itemize}

\end{theo}
The lemmata and propositions in this section constitute the proof. To show the statement of Theorem~\ref{theo_unifbounds} we proceed stepwise, first discussing the different zeroth derivatives and then proceeding to the higher ones. 
\begin{rema}
\label{rema_assumpotentialforglobalexist}
	Throughout this section, we assume that the potential~$\potV:\RR\rightarrow\RR$ is non-negative and~$C^\infty$. 
	In particular, this implies that~$\potV$ as well as its first and second derivatives~$V'=d\potV/d\scalphi$ and~$V''=d^2\potV/d\scalphi^2$ do not become unbounded for a bounded scalar field~$\scalphi$.
\end{rema}

Inspired by~\cite[eq.~(42)]{andreassonrendallweaver_existCMCT2Vlasov}, we define the following energy:
\begin{align}
\label{eqn_defiEnew}
	\Enew=\int_{\SS^1}\alpha^{{-}1/2} &
	\left[ (P_t+\frac1t)^2+\alpha P_\theta^2 + e^{2P}(Q_t^2+\alpha Q_\theta^2) + 2(\scalphi_t^2+\alpha\scalphi_\theta^2) \right.\\
	 & \left.  + \quotJ7 + \quotKQJ7 + 4t^{{-}1/2}e^{\lambda/2}\potV \right] d\theta.	
\end{align}
We remark at this point that all terms in the integrand are non-negative due to the assumption on the potential~$\potV$. A direct comparison with the evolution equations, mainly equations~\eqref{eqn_secondevollambda} and~\eqref{eqn_timeevolalpha}, shows that it is equivalent to write
\begin{equation}
\label{eqn_Enewequiv}
	\Enew=\int_{\SS^1}\alpha^{{-}1/2}t^{{-}1}\left[\lambda_t-2\frac{\alpha_t}{\alpha} +\frac1t +2P_t\right]d\theta.
\end{equation}
\begin{lemm}
\label{lemm_Enewbound}
	\consevolequsinterval
	Then, independently of the sign of the potential~$\potV$, the energy~$\Enew$ is monotone decreasing on~$I$, \ie
	\begin{equation}
		\Enew(t)\le \Enew(t_0),
	\end{equation}
	for all~$t\ge t_0$, where~$t,t_0\in I$. In particular, $\Enew$ is bounded towards the future.
	
	If additionally the potential~$\potV$ is non-negative, then there exists a constant~$C$, depending only on the initial data at time~$t_0\in I$, such that
	\begin{equation}
		\int_{\SS^1} \left[ P_t^2+\alpha P_\theta^2 + e^{2P}(Q_t^2+\alpha Q_\theta^2) + 2(\scalphi_t^2+\alpha\scalphi_\theta^2) \right] d\theta \le C ,
	\end{equation}
	for all~$t\ge t_0$, where~$\in I$.
\end{lemm}
\begin{proof}
	Differentiating the integrand of~$\Enew$ with respect to time, one finds that
\begin{align}
	\frac{d}{dt}(\Enew(t))={-}\int_{\SS^1} 2t^{{-}1}\alpha^{{-}1/2} & 
	\left[(P_t+\frac1t)^2 + e^{2P}\alpha Q_\theta^2 + 2\scalphi_t^2\right.\\
	& \left. +\quotJ7 +2\quotKQJ7 \right] d\theta,
\end{align}
which is non-positive. This yields the monotonocity statement.

Considering the second estimate in the statement, for all but the first term in the integral this follows directly from the energy estimate together with Lemma~\ref{lemm_estimatealphaVnonnegative} stating that~$\alpha$ is monotone decaying. Instead of obtaining an estimate for~$P_t$ directly, we find that
\begin{equation}
	\int_{\SS^1} (P_t+\frac1t)^2 d\theta \le C.
\end{equation}
This gives
\begin{equation}
	\int_{\SS^1}\absval{P_t+\frac1t}d\theta
	\le C,
\end{equation}
the latter implying that even
\begin{equation}
	\int_{\SS^1} \absval{P_t}d\theta \le C,
\end{equation}
and can be combined to
\begin{equation}
	 \int_{\SS^1} P_t^2 d\theta=\int_{\SS^1} \left[(P_t+\frac1t)^2-\frac2tP_t-\frac1{t^2}\right]d\theta \le C.
\end{equation}
This concludes the proof.
\end{proof}
From Lemma~\ref{lemm_estimatealphaVnonnegative}, we know the following: If we assume a non-negative potential, then~$\alpha$ is monotone decaying towards the future, and therefore, on every compact set~$I=[t_0,t]$, is uniformly bounded from above by a constant depending only on the initial data at time~$t_0$. Further, $\alpha$ is bounded from below by zero.
We will need a stronger lower bound and prove an integral version now as well as a pointwise version further down in Lemma~\ref{lemm_alphalowerbound}.
\begin{lemm}
\label{lemm_alphaintegrallowerbound}
	\consevolequscompinterval
	Then~$\int_{\SS^1}\alpha^{{-}1/2}d\theta$ is uniformly bounded from above by a constant \dependconstant.
\end{lemm}
\begin{proof}
	We compute
	\begin{equation}
	      \partt\alpha^{{-}1/2}={-}\frac12\alpha^{{-}1/2}\frac{\alpha_t}{\alpha},
	\end{equation}
	and notice from the evolution equation~\eqref{eqn_timeevolalpha} for~$\alpha$ that this
	is bounded from above by~$t/2$ times the integrand of~$\Enew$ defined in~\eqref{eqn_Enewequiv}.
	Consequently, setting~$D\coloneqq \max_{\theta\in\SS^1}(\alpha^{{-}1/2}(t_0,\theta))$ and integrating over the integral~$[t_0,t]$ as well as over~$\SS^1$, we find
	\begin{equation}
	      \int_{\SS^1}\alpha^{{-}1/2}(t,\theta)d\theta \le \int_{t_0}^t \frac s2 \Enew(s) ds +\int_{\SS^1}\alpha^{{-}1/2}(t_0,\theta)d\theta
	      \le\Enew(t_0)\frac{t^2-t_0^2}4 + 2\pi D,
      \end{equation}
      which is uniformly bounded as claimed.
\end{proof}

\begin{lemm}
\label{lemm_unifboundPQscalphi}
	\consevolequscompinterval
	Then~$P$, $Q$ and~$\scalphi$ are uniformly bounded by constants \dependconstant.
\end{lemm}
\begin{proof}
	We give the proof for~$P$ and leave its adaptation to~$Q$ and~$\scalphi$ to the reader. We can estimate the derivative of~$P+\ln t$ via
	\begin{equation}
		\absval{\partial_t\mean{P+\ln t}}
		\le\frac1{2\pi}\int_{\SS^1}\absval{P_t+\frac1t}d\theta 
		\le \frac1{\sqrt{2\pi}}(\int_{\SS^1}(P_t+\frac1t)^2d\theta)^{1/2},
	\end{equation}
	which is uniformly bounded due to Lemma~\ref{lemm_Enewbound}.
	As a consequence, the same boundedness holds true for~$\mean{P+\ln t}$, and consequently also for~$\mean P$. As further
	\begin{equation}
		\int_{\SS^1} \absval{P_\theta}d\theta 
		\le \frac1{2\pi}(\int_{\SS^1} \alpha^{1/2} P_\theta^2d\theta)^{1/2} (\int_{\SS^1} \alpha^{{-}1/2}d\theta)^{1/2},
	\end{equation}
	the integral bound on~$\alpha^{{-}1/2}$ found in Lemma~\ref{lemm_alphaintegrallowerbound} and again Lemma~\ref{lemm_Enewbound} imply the statement for~$P$.
\end{proof}

\begin{lemm}
\label{lemm_unifboundlambda}
	\consevolequscompinterval
	Assume further that the potential~$\potV$ is non-negative.
	Then~$\lambda$ is uniformly bounded by a constant \dependconstant.
\end{lemm}
\begin{proof}
	From the evolution equation~\eqref{eqn_spaceevollambda}, we conclude that
	\begin{align}
		\absval{\lambda_\theta}={}&\absval{2t(P_tP_\theta +e^{2P}Q_tQ_\theta)+4t\scalphi_t\scalphi_\theta}\\
		  \le{}&t\alpha^{{-}1/2}(P_t^2+\alpha P_\theta^2 + e^{2P}(Q_t^2+\alpha Q_\theta^2) + 2(\scalphi_t^2+\alpha \scalphi_\theta^2)).
	\end{align}
	Lemma~\ref{lemm_alphaintegrallowerbound} and Lemma~\ref{lemm_Enewbound} thus yield that~$\absval{\lambda(t,\theta_1)-\lambda(t,\theta_2)}$ is uniformly bounded on compact time intervals.
	Further
	\begin{equation}
		\absval{\partt\mean\lambda}
		\le\frac1{2\pi}\int_{\SS^1}\absval{\lambda_t},
	\end{equation}
	which, due to evolution equation~\eqref{eqn_timeevollambda}, is bounded via the energy~$\Enew$. Due to monotonicity of this energy, Lemma~\ref{lemm_Enewbound}, and the upper bound on~$\alpha$ from Lemma~\ref{lemm_estimatealphaVnonnegative}, this implies that~$\mean\lambda$ is uniformly bounded on compact time intervals, and combining this statement with the bounded difference of~$\lambda$ at two points~$\theta_1,\theta_2$ implies the statement.
\end{proof}

As a direct consequence of the previous lemmata showing uniform bounds for~$\alpha$ (from above, and an integral bound from below), $\lambda$, $P$, $Q$ and~$\scalphi$, we conclude the following:
\begin{coro}
\label{coro_unifboundquotients}
	\consevolequscompinterval
	\assumptnonnegatpotential
	Then quotients of the form
	\begin{equation}
		\frac{e^{\lambda/2-P}J^2}{t^{\mu}}, \qquad \frac{e^{\lambda/2+P}(K-QJ)^2}{t^{\mu}},
	\end{equation}
	$\mu\in\RR$, as well as the potential~$\potV$ are uniformly bounded by constants \dependconstant.
\end{coro}
Strictly speaking, at this point it is enough to assume~$C^1$ or even Lipschitz regularity for the potential, but from here on out we make the smoothness assumption we also use in the global existence and uniqueness statement.

We can now improve the integral lower bound on~$\alpha$ to a pointwise one.
\begin{lemm}
\label{lemm_alphalowerbound}
	\consevolequscompinterval
	\assumptnonnegatpotential
	Then~$\alpha^{{-}1}$ is uniformly bounded from above, or equivalently~$\alpha$ is uniformly bounded away from zero, by constants \smoothfctdependconstant.
\end{lemm}
\begin{proof}
	We consider~$\ln \alpha$ and conclude from the evolution equation~\eqref{eqn_timeevolalpha} that its time derivative is non-positive and uniformly bounded on compact time intervals due to the uniform boundedness of all zeroth derivatives, which we have proven in the earlier statements in this section. As a consequence, $\ln\alpha$ is uniformly bounded by a constant \dependconstant.
	By the proporties of the exponential function, the statement follows.
\end{proof}

Having found uniform bounds for all non-differentiated variables, we now turn to their first derivatives.
\begin{lemm}
\label{lemm_unifboundtimederivalpha}
	\consevolequscompinterval
	\assumptnonnegatpotential
	Then~$\alpha_t/\alpha$ and~$\alpha_t$ are uniformly bounded by constants \dependconstant.
\end{lemm}
\begin{proof}
	This is a direct consequence of the evolution equation~\eqref{eqn_timeevolalpha} and boundedness of all non-differentiated variables. Note that we do not make use of the lower bound of~$\alpha$.
\end{proof}

For the first derivatives of the variables~$P$, $Q$ and~$\scalphi$, we use a lightcone integration ansatz. Similar approaches have been used for example in~\cite{andreassonrendallweaver_existCMCT2Vlasov} or~\cite{andreasson_globalfoliatgowdysymm} (both times in Part 2 of Section 4). 
To this end, we consider the quantities~$\Afctpm$ and~$\Bfctpm$ from equations~\eqref{eqn_defiAfctpm} and~\eqref{eqn_defiBfctpm} whose lightcone derivatives we can estimate using the boundedness statements for the zeroth derivatives. We integrate these lightcone derivatives along characteristic curves between the~$t_0$- and the~$t$-timeslice. Grönwall's lemma then yields uniform bounds for~$\Afctpm$ and~$\Bfctpm$, which imply the requested bounds on the derivatives of~$P$, $Q$ and~$\scalphi$.

\begin{prop}
\label{prop_estimfirstderivPQ}
	\consevolequscompinterval
	\assumptnonnegatpotential
	Then~$P_t$, $P_\theta$, $Q_t$ and~$Q_\theta$ are uniformly bounded, by constants \smoothfctdependconstant.
\end{prop}
\begin{prop}
\label{prop_estimfirstderivscalphi}
	\consevolequscompinterval
	\assumptnonnegatpotential
	Then~$\scalphi_t$ and~$\scalphi_\theta$ are uniformly bounded, by constants \smoothfctdependconstant.
\end{prop}

\begin{proof}[Proof of Proposition~\ref{prop_estimfirstderivPQ}]
	We consider the function~$\Afctmp$ defined in~\eqref{eqn_defiAfctpm}, whose lightcone derivative has been computed in Lemma~\ref{lemm_Afctlightconederiv}. All terms appearing in the second line of that statement can be estimated by~$C\sqrt{\Afctmp}$, due to uniform boundedness of the quotients given via Corollary~\ref{coro_unifboundquotients} and Lemma~\ref{lemm_unifboundPQscalphi}. For the terms in the first line in Lemma~\ref{lemm_Afctlightconederiv}, a short computation of the individual terms of the function~$\Afctmp$ gives
	\begin{align}
		\mp\frac2t\alpha^{1/2}(P_\mp P_\theta +e^{2P}Q_\mp Q_\theta)
		  ={}&\frac{1}{2t}(\Afctmp-\Afctpm)+\frac2t\alpha(P_\theta^2+e^{2P} Q_\theta^2)\\
		  \le{}&\frac1t \Afctmp +\frac1{2t}(\Afct_++\Afct_-).
	\end{align}
	Further, $\alpha_t/\alpha$ is uniformly bounded by Lemma~\ref{lemm_unifboundtimederivalpha}. Consequently, we can estimate
	\begin{equation}
	\label{eqn_auxilestimintegrderivAfctpm}
		  \absval{\partpm\Afctmp}
		  \le C\Afctmp + C\sqrt{\Afctmp} +C(\Afct_++\Afct_-)
		  \le C+ C(\Afct_++\Afct_-),
	\end{equation}
	where in the second estimate we have used~$\sqrt x\le 1+x$.

	We now set
	\begin{equation}
		G(t,\theta)=P_{t}^2+\alpha P_{\theta}^2+e^{2P}Q_{t}+e^{2P}\alpha Q_{\theta}^2=\frac12(\Afct_++\Afct_-)
	\end{equation}
	and fix a point~$(t,\theta)$. We further set~$c_\pm$ to be two characteristic curves in the spacetime, starting at some point on the~$t_0$-level and ending in the point~$(t,\theta)$, such that their velocity vector is~$\partpm$.
	In particular, the curve parameter~$s$ can be chosen such that~$c_\pm(s)$ is contained in the set~$\{s\}\times \TT^3$.

	We can rewrite~$G$ by integrating along these curves, in fact
	\begin{align}
		G(t,\theta)={}&\frac12(\Afct_++\Afct_-)(t,\theta)\\
		  ={}&\frac12\Afct_+(c_-(t_0))+\frac12\int_{t_0}^{t}\partial_-\Afct_+(c_-(s))ds\\
		      &{}+\frac12\Afct_-(c_+(t_0))+\frac12\int_{t_0}^{t}\partial_+\Afct_-(c_+(s))ds.
	\end{align}
	Applying estimate~\eqref{eqn_auxilestimintegrderivAfctpm}, we find that
	\begin{equation}
		\sup_{\theta\in\SS^1}G\le C+ \int_{t_0}^{t_1}(C\sup_{\theta\in\SS^1}G)ds.
	\end{equation}
	Grönwall's lemma yields a uniform bound on~$\sup_\theta G$ and hence on~$G$, which implies uniform bounds on~$P_t$, $e^PQ_t$, $\alpha^{1/2}P_\theta$ and~$\alpha^{1/2}e^PQ_\theta$. The lower bound on~$\alpha$ found in Lemma~\ref{lemm_alphalowerbound} and the uniform bound on~$P$, see Lemma~\ref{lemm_unifboundPQscalphi}, concludes the proof.
\end{proof}

\begin{proof}[Proof of Proposition~\ref{prop_estimfirstderivscalphi}]
	A short computation of the individual terms of the function~$\Bfctpm$ defined in~\eqref{eqn_defiBfctpm} reveals that
	\begin{align}
		\mp \frac2t\alpha^{1/2}\scalphi_\theta\scalphi_\mp
		  ={}&\frac{1}{2t}(\Bfctmp-\Bfctpm) +\frac2t\alpha\scalphi_\theta^2\\
		  \le{}&\frac1t\Bfctmp +\frac1{2t}(\Bfct_++\Bfct_-).
	\end{align}
	Combining this estimate with Lemma~\ref{lemm_Bfctlightconederiv}, we can estimate the lightcone derivative of~$\Bfctpm$:
	\begin{equation}
	\label{eqn_lightestimBfctmp}
		\partpm\Bfctmp\le {-}\frac1t \Bfctmp +\frac1{2t}(\Bfct_++\Bfct_-) +\frac{\alpha_t}{\alpha}\scalphi_\mp^2-2t^{{-}1/2}e^{\lambda/2}\potV'\scalphi_\mp.
	\end{equation}
	
	The first derivative of the smooth function~$\potV$ has uniform bounds on the compact set of possible~$\scalphi$-values, see Remark~\ref{rema_assumpotentialforglobalexist}, and using the uniform bounds for all zeroth derivatives, we can estimate the last term in~\eqref{eqn_lightestimBfctmp} via
	\begin{equation}
		\absval{2t^{{-}1/2}e^{\lambda/2}\potV'\scalphi_\mp}\le C\absval{\scalphi_\mp}\le C+C\scalphi_\mp^2.
	\end{equation}
	Using Lemma~\ref{lemm_unifboundtimederivalpha} to estimate the term~$\alpha_t/\alpha$, we therefore find that inequality~\eqref{eqn_lightestimBfctmp} implies
	\begin{equation}
		\absval{\partpm\Bfctmp}\le C+ C(\Bfct_++\Bfct_-).
	\end{equation}
	This is the same estimate as the one we had available for~$\Afctmp$ in the previous proof, see~\eqref{eqn_auxilestimintegrderivAfctpm}.
	We set
	\begin{equation}
		H(t,\theta)=\scalphi_t^2+\alpha\scalphi_\theta^2=\frac12(\Bfct_++\Bfct_-)
	\end{equation}
	and integrate the lightcone derivatives of~$\Bfctpm$ along characteristic curves~$c_\pm$ defined as in the previous proof. Using the same steps as in the previous proof, we show uniform boundedness of~$H$, and from this obtain the statement.
\end{proof}

As a direct consequence of Proposition~\ref{prop_estimfirstderivscalphi} giving uniform bounds on the first derivatives of~$\scalphi$, we find the following corollary, see also Remark~\ref{rema_assumpotentialforglobalexist}.

\begin{coro}
\label{coro_boundsremainingfirstderiv}
	\consevolequscompinterval
	\assumptnonnegatpotential
	Then the first derivatives~$\potV'=d\potV/d\scalphi$, $\partial_t\potV=\potV'\scalphi_t$ and~$\partial_\theta\potV=\potV'\scalphi_\theta$ of the potential are uniformly bounded by constants \smoothfctdependconstant.
\end{coro}
\begin{lemm}
	\consevolequscompinterval
	\assumptnonnegatpotential
	Then~$\lambda_t$ and~$\lambda_\theta$ are uniformly bounded by constants \smoothfctdependconstant.
\end{lemm}
\begin{proof}
	This follows directly from the evolution equations, as all quantities appearing on the \rhs\ of equations~\eqref{eqn_timeevollambda} and~\eqref{eqn_spaceevollambda} are uniformly bounded by the boundedness statements we have obtained so far.
\end{proof}

\begin{lemm}
	\consevolequscompinterval
	\assumptnonnegatpotential
	Then~$\alpha_\theta$ is uniformly bounded by constants \smoothfctdependconstant.
\end{lemm}
\begin{proof}
	Using the evolution~\eqref{eqn_timeevolalpha}, we find that the~$\theta$-derivative of~$\alpha_t$ is
	\begin{equation}
		\alpha_{\theta t}=\frac{\alpha_t}{\alpha}\alpha_\theta + \alpha\parttheta(\frac{e^{\lambda/2-P}J^2}{t^{5/2}} + \frac{e^{\lambda/2+P}(K-QJ)^2}{t^{5/2}} + 4t^{1/2}e^{\lambda/2}\potV(\scalphi)).
	\end{equation}
	This equation is of the form
	\begin{equation}
		\partt\alpha_\theta=f\alpha_\theta +g,
	\end{equation}
	where the functions~$f$ and~$g$ are uniformly bounded by constants \smoothfctdependconstant, due to the results we have achieved in this section so far. Integration implies the statement.

\end{proof}
At this point, we have proven that on compact time intervals all first order derivatives, both with respect to time~$t$ and space~$\theta$, of all variables appearing in the evolution equations, are uniformly bounded by constants \smoothfctdependconstant. We continue with the second derivatives, but are only interested in~$P$, $Q$ and~$\scalphi$, as these are the only variables whose second derivatives we need to control in order to conclude global existence in Section~\ref{section_globalexist}. For the proof, we first provide bounds on the second derivatives of~$\alpha$.

\begin{lemm}
	\consevolequscompinterval
	\assumptnonnegatpotential
	Then~$\alpha_{tt}$ and~$\alpha_{t\theta}$ are uniformly bounded by constants \smoothfctdependconstant.
\end{lemm}
\begin{proof}
	The~$\theta$- and~$t$-derivative of the evolution equation~\eqref{eqn_timeevolalpha} contain only zeroth and first derivative terms. Those are uniformly bounded due to the statement we have obtained so far.
\end{proof}

\begin{prop}
\label{prop_estimsecondderivPQ}
	\consevolequscompinterval
	\assumptnonnegatpotential
	Then the second derivatives~$P_{tt}$, $P_{t\theta}$ and~$P_{\theta\theta}$ of~$P$, and equivalently for~$Q$, are uniformly bounded by constants \smoothfctdependconstant.
\end{prop}
\begin{proof}
	We start by rewriting the evolution equations~\eqref{eqn_secondevolP} and~\eqref{eqn_secondevolQ} as follows:
	\begin{equation}
	\label{eqn_secondevolPrewritten}
		P_{tt}-\alpha P_{\theta\theta}={-}\frac 1tP_t +\frac{\alpha_t}{2\alpha}P_t +\frac{\alpha_\theta}{2}P_\theta +e^{2P}(Q_t^2-\alpha Q_\theta^2)+\frac{e^{\lambda/2-P}J^2}{2t^{7/2}}-\frac{e^{\lambda/2+P}(K-QJ)^2}{2t^{7/2}}
	\end{equation}
	and 
	\begin{equation}
	\label{eqn_secondevolQrewritten}
		Q_{tt}-\alpha Q_{\theta\theta}={-}\frac 1tQ_t +\frac{\alpha_t}{2\alpha}Q_t +\frac{\alpha_\theta}{2}Q_\theta -2(P_tQ_t-\alpha P_\theta Q_\theta)+\frac{e^{\lambda/2-P}J(K-QJ)}{t^{7/2}}.
	\end{equation}
	We consider 
	\begin{equation}
		\Afcttwopm=(P_{tt}\pm\alpha^{1/2}P_{t\theta})^2+ e^{2P}(Q_{tt}\pm\alpha^{1/2}Q_{t\theta})^2
	\end{equation}
	and wish to estimate the lightcone derivative~$\partpm \Afcttwomp$. In order to do so, we notice that
	\begin{equation}
		\partpm(P_{tt}\mp\alpha^{1/2}P_{t\theta})=\partt(P_{tt}-\alpha P_{\theta\theta})+\alpha_tP_{\theta\theta} \mp\frac{\alpha_t}{2\alpha}\alpha^{1/2}P_{t\theta}-\frac{\alpha_\theta}{2} P_{t\theta}.
	\end{equation}
	The term with~$P_{\theta\theta}$ can be replaced by the evolution equation~\eqref{eqn_secondevolPrewritten}, and the first bracket can be computed using the time derivative of the same equation. This yields an expression which contains only derivatives up to order two, and none where~$P$ or~$Q$ is differentiated with respect to~$\theta$ twice.
	Moreover, the second derivatives of~$P$ and~$Q$ occur linearly.
	Equivalently, we can treat~$\partpm(e^{2P}(Q_{tt}\mp\alpha^{1/2}Q_{t\theta}))$. 
	We now combine the two expressions, using the uniform bounds of all zeroth and first derivatives as well as of~$\alpha_{tt}$ and~$\alpha_{t\theta}$ which we have achieved in this section so far. This yields
	\begin{equation}
		  \absval{\partpm\Afcttwomp}
		  \le C+C(\Afcttwoplus+\Afcttwominus).
	\end{equation}
	This estimate coincides with the one we found for~$\Afctmp$ in the proof of Proposition~\ref{prop_estimfirstderivPQ}, see~\eqref{eqn_auxilestimintegrderivAfctpm}, and the same chain of arguments as we employed there shows that
	\begin{equation}
		P_{tt}^2+\alpha P_{t\theta}^2+e^{2P}Q_{tt}+e^{2P}\alpha Q_{t\theta}^2
	\end{equation}
	is uniformly bounded on compact time intervals. Combining this with Lemma~\ref{lemm_alphalowerbound} providing a lower bound on~$\alpha$, as well as equations~\eqref{eqn_secondevolPrewritten} and~\eqref{eqn_secondevolQrewritten}, we also find uniform bounds for~$P_{t\theta}^2$ and $Q_{t\theta}^2$ as well as~$P_{\theta\theta}^2$ and~$Q_{\theta\theta}^2$, which concludes the proof.
\end{proof}
\begin{prop}
\label{prop_estimsecondderivscalphi}
	\consevolequscompinterval
	\assumptnonnegatpotential
	Then the second derivatives~$\scalphi_{tt}$, $\scalphi_{t\theta}$ and~$\scalphi_{\theta\theta}$ are uniformly bounded by constants \smoothfctdependconstant.
\end{prop}
\begin{proof}
	We consider
	\begin{equation}
		\Bfcttwopm=(\scalphi_{tt}\pm\alpha^{1/2}\scalphi_{t\theta})^2
	\end{equation}
	and, similar to the approach in the previous proof, compute
	\begin{equation}
		\partpm(\scalphi_{tt}\mp\alpha^{1/2}\scalphi_{t\theta}) =\partt(\scalphi_{tt}-\alpha \scalphi_{\theta\theta})+\alpha_t\scalphi_{\theta\theta} \mp\frac{\alpha_t}{2\alpha}\alpha^{1/2}\scalphi_{t\theta}-\frac{\alpha_\theta}{2} \scalphi_{t\theta}.
	\end{equation}
	Using the equation of motion~\eqref{eqn_eqnofmotionscalfieldwithT2} as well as its time derivative, this can be refomulated into an expression with derivatives only up to second order, and where~$\scalphi$ is not differentiated with respect to~$\theta$ twice.
	Moreover, the second order derivatives of~$\scalphi$ appear linearly.
	We can therefore proceed to estimate
	\begin{equation}
		  \absval{\partpm\Bfcttwomp}
		  \le C+C(\Bfcttwoplus+\Bfcttwominus).
	\end{equation}
	In this last step, we applied the assumption on the potential in order to bound the time derivative of the term with~$V'$, appearing on the \lhs\ of the equation of motion, see also Remark~\ref{rema_assumpotentialforglobalexist}.
	As in the previous proof and that of Proposition~\ref{prop_estimfirstderivPQ}, we conclude from this estimate that
	\begin{equation}
		\scalphi_{tt}^2+\alpha \scalphi_{t\theta}^2
	\end{equation}
	is uniformly bounded on compact time intervals. The lower bound on~$\alpha$ from Lemma~\ref{lemm_alphalowerbound} and the equation of motion, equation~\eqref{eqn_eqnofmotionscalfieldwithT2}, imply that the same holds for~$\scalphi_{t\theta}^2$ and~$\scalphi_{\theta\theta}^2$.
\end{proof}
As a direct consequence, this even yields boundedness of the second derivatives of the potential.
\begin{coro}
	\consevolequscompinterval
	\assumptnonnegatpotential
	Then the second derivatives~$\potV''=d^2\potV/d\scalphi^2$, $\partt^2\potV$, $\partt\parttheta\potV$ and~$\parttheta^2\potV$ of the potential are uniformly bounded by constants \smoothfctdependconstant.
\end{coro}

\begin{rema}
\label{rema_unifboundsextendhigherderiv}
	We end our discussion at this point, as we have achieved uniform bounds on all quantities we need for the proof of global existence in the next section. However, we can continue to even higher derivatives:
	In the same spirit as we have treated the first derivatives of~$\lambda$, we can now go on to treat its second order derivatives. 
	With this, we find uniform bounds on all second derivatives. Continuing inductively, using the same steps as for the second derivatives we have shown here, we can extend this discussion to even higher derivatives. 
	
	In detail, we can show the following: For all variables in the evolution equations~\eqref{eqn_timeevolalpha}--\eqref{eqn_eqnofmotionscalfieldwithT2}, all derivatives up to order~$N$ are uniformly bounded on compact time intervals~$I=[t_0,t]$,
	by constants \smoothfctdependconstant.
\end{rema}

\section{Global existence and uniqueness}
\label{section_globalexist}

In Theorem~\ref{theo_unifbounds} in the previous section, we have shown that solutions to the evolution equations~\eqref{eqn_timeevolalpha}--\eqref{eqn_eqnofmotionscalfieldwithT2} are uniformly bounded on compact time intervals, by constants \smoothfctdependconstant. This boundedness was shown to hold up to first and, for certain functions, even second derivative.
We now relate our set of evolutions equations to a general existence and uniqueness result for partial differential equations from~\cite{majda_comprfluidsconservlaws} and prove global existence and uniqueness of solutions.

\begin{proof}[Proof of Theorem~\ref{theo_globalexistence}]
	After giving the evolution equations~\eqref{eqn_timeevolalpha}--\eqref{eqn_eqnofmotionscalfieldwithT2}, we have shown on page~\pageref{eqn_evoleqnisconstraint} that equation~\eqref{eqn_secondevollambda} is a consequence of the remaining ones, and argued that equation~\eqref{eqn_spaceevollambda} for~$\lambda_\theta$ can be considered a constraint equation. For these reasons, one easily sees that the set of evolution equations~\eqref{eqn_timeevolalpha}--\eqref{eqn_eqnofmotionscalfieldwithT2} is equivalent to the time evolution of the vector
	\begin{equation}
	\label{eqn_vectorglobalexist}
		u\coloneqq(\alpha,\lambda,P,Q,\scalphi,\frac{P_t}{\sqrt\alpha},P_\theta,\frac{Q_t}{\sqrt\alpha},Q_\theta,\frac{\scalphi_t}{\sqrt\alpha},\scalphi_\theta)^T.
	\end{equation}
	
	Now, one checks that the time evolution of the vector~\eqref{eqn_vectorglobalexist} admits the structure of a symmetric hyperbolic system as defined in~\cite[(2.1a),(2.1b),(1.17)]{majda_comprfluidsconservlaws}. To this end, one computes~$\partial u/\partial t$ and realises that all terms containing derivatives of components of~$u$ can be collected in one expression, namely
	\begin{equation}
		F \cdot \frac{\partial u}{\partial \theta},
	\end{equation}
	where~$F$ is the~$11$ by~$11$ matrix whose only non-zero entries are the vector
	\begin{equation}
		v^T\coloneqq\frac{1}{2\sqrt\alpha}(0,0,0,0,0,P_\theta,\frac{P_t}{\sqrt\alpha},Q_\theta,\frac{Q_t}{\sqrt\alpha},\scalphi_\theta,\frac{\scalphi_t}{\sqrt\alpha})^T
	\end{equation}
	as first column and three permutation matrices
	\begin{equation}
		\sqrt\alpha\begin{pmatrix}
		           	0 & 1\\ 1 & 0
		           \end{pmatrix}
	\end{equation}
	on the diagonal of the lower right~$6$ by~$6$ submatrix. Defining now a matrix~$\tilde A$ whose first column and first row are set to 
	\begin{equation}
		\frac{1}{2\alpha}(2\alpha+2v^Tv,0,0,0,0,\frac{P_t}{\sqrt\alpha},P_\theta,\frac{Q_t}{\sqrt\alpha},Q_\theta,\frac{\scalphi_t}{\sqrt\alpha},\scalphi_\theta)^T,
	\end{equation}
	all diagonal elements apart from the first one are set to~$1$, and all remaining components vanish, we find that
	\begin{equation}
		\tilde A \cdot F
	\end{equation}
	is a symmetric matrix. Further, $\tilde A$ is symmetric, positive definite, and bounded from above and below by multiples of the identity matrix. Consequently, our system of equations satisfies the properties~\cite[(1.17)]{majda_comprfluidsconservlaws}.

	Given smooth initial data~$u_0$, we can therefore apply Theorems~2.1,
	or alternatively Corollary~1, in~\cite{majda_comprfluidsconservlaws} to this system. This yields uniqueness of solutions in a small time interval. As all components of the vector~\eqref{eqn_vectorglobalexist} are uniformly bounded up to their~first derivative due to Theorem~\ref{theo_unifbounds}, and as~$\SS^1$ is a compact set, the continuation criterion in Theorem~2.2 on page~31--32, or alternatively Corollary 2, in the same reference implies that the maximal existence interval has no upper bound, as well as (tacitly) gives global uniqueness. This concludes the proof.
\end{proof}

\section{Constant potential}
\label{section_constpotV}

In this section, we assume the potential to be constant and positive, $\potV=\constpot>0$, and show estimates for the individual variables. This leads us to a proof of the \nohairconj\ in this setting. The approach is inspired by that taken in~\cite{andreassonringstrom_cosmicnohairT3gowdyeinsteinvlasov} for the case of Vlasov matter with positive cosmological constant.
Some of the statements and proofs have a counterpart in that reference. However, the details differ, and all statements about the scalar field~$\scalphi$ are new. We therefore provide the intermediate steps, though sometimes in an abbreviated form, if the arguments are very close to those of~\cite{andreassonringstrom_cosmicnohairT3gowdyeinsteinvlasov}. For the convenience of the reader familiar with that paper, the structure and notation in the present paper is similar to theirs.

\begin{rema}
	In the previous sections, we have found that solutions to equations~\eqref{eqn_timeevolalpha}--\eqref{eqn_eqnofmotionscalfieldwithT2} have a maximal existence interval which extends to~$+\infty$, this is the statement of Theorem~\ref{theo_globalexistence}. 
	The resulting~$\TT^2$-symmetric metrics as in equation~\eqref{eqn_metricT2} are therefore defined on~$I\times\TT^3$, for~$I=(t_0,\infty)$, $t_0>0$. Throughout this section, we denote solutions of this form as~\textit{future global}. 
	
	We further set~$t_1=t_0+2$ and use it without additional explanation. This choice is made to ensure that~$\ln t_1$ positive and bounded away from zero.
\end{rema}

\begin{rema}
	We know that solutions to equations~\eqref{eqn_timeevolalpha}--\eqref{eqn_eqnofmotionscalfieldwithT2} are future global, and further that a constant potential~$\constpot$ is in particular bounded towards the future. In~$\TT^3$-Gowdy symmetry, we can therefore conclude from Lemma~\ref{lemm_lambdaasympt} that the resulting metric satisfies~$\lambda$-asymptotics with constant~$\constpot$, see Definition~\ref{defi_lambdaasympt}. Throughout this section, we assume that the same asymptotic behaviour holds true even for the general~$\TT^2$-symmetry, \ie we assume that there is, for every~$\eps>0$, a~$T>t_0$ such that
	\begin{equation}
		\lambda(t,\theta)\ge{-}3\ln t +2\ln(\frac{3}{4\constpot})-\eps,
	\end{equation}
	for all~$(t,\theta)\in[T,\infty)\times\SS^1$.

	In our setting of a constant potential, we can also make use of Lemma~\ref{lemm_estimatealphaVconstant}, providing us with an upper bound for~$\alpha$:
	\begin{equation}
		\alpha(t,\theta)\le Ct^{-3}
	\end{equation}
	for all $(t,\theta)\in [t_1,\infty)\times\SS^1$.
\end{rema}

We define the energy
\begin{equation}
\label{eqn_defiEbas}
\begin{subaligned}
	\Ebas={}&\int_{\SS^1}t\alpha^{{-}1/2}(\lambda_t-2\frac{\alpha_t}{\alpha}-4t^{1/2}e^{\lambda/2}\constpot)d\theta\\
	  ={}&\int_{\SS^1}\left(t^2\alpha^{{-}1/2}\left[P_t^2+\alpha P_\theta^2+e^{2P}(Q_t^2+\alpha Q_\theta^2)+2(\scalphi_t^2+\alpha\scalphi_\theta^2)\right]\right.\\
	    &\quad{}+ \left.t^{{-}3/2}\alpha^{{-}1/2}(e^{\lambda/2-P}J^2 + e^{\lambda/2+P}(K-QJ)^2)\right)d\theta,
\end{subaligned}
\end{equation}
which is constructed in the same spirit as the energy defined in~\cite[eq.~(85)]{andreassonringstrom_cosmicnohairT3gowdyeinsteinvlasov}. See also their discussion after the energy's definition for its advantages compared to other energies.
In case of~$\lambda$-asymptotics, we can estimate the asymptotic behaviour of~$\Ebas$ towards the future.
\begin{lemm}
\label{lemm_ebasestimatewitha}
	\consconstVandTtwoandlambdaasympt
	Then, for every~$a>1/2$ there is a constant~$C_a>0$ such that
	\begin{equation}
		\Ebas(t)\le C_a t^a,
	\end{equation}
	for all~$t\ge t_1$.
\end{lemm}
\begin{proof}
	The proof is similar to that of~\cite[Lemma~45]{andreassonringstrom_cosmicnohairT3gowdyeinsteinvlasov}: 
	Using the evolution equation~\eqref{eqn_secondevollambda}, we compute the time derivative of the integrand in the first expression of the energy, equation~\eqref{eqn_defiEbas}:
	\begin{align}
		\partt&\left[t\alpha^{{-}1/2}(\lambda_t-2\frac{\alpha_t}{\alpha}-4t^{1/2}e^{\lambda/2}\constpot)\right]\\
		  &\quad= \parttheta(t\alpha^{1/2}\lambda_\theta)+2t\alpha^{1/2}(P_\theta^2+e^{2P}Q_\theta^2+2\scalphi_\theta^2)\\
		  &\qquad-\frac32t\alpha^{{-}1/2}(\quotJ{7}+\quotKQJ{7})\\
		  &\qquad-(P_t^2+\alpha P_\theta^2+e^{2P}(Q_t+\alpha Q_\theta^2)+2(\scalphi_t^2+\alpha\scalphi_\theta^2)) (2t^{5/2}\alpha^{{-}1/2}e^{\lambda/2}\constpot).
	\end{align}
	From this, we can estimate
	\begin{equation}
		\frac{d\Ebas}{dt}\le \int_{\SS^1} 2t\alpha^{1/2}(P_\theta^2+e^{2P}Q_\theta^2+2\scalphi_\theta^2)d\theta 
		    - \int_{\SS^1}2t^{5/2}\alpha^{1/2}e^{\lambda/2}\constpot (P_\theta^2+e^{2P}Q_\theta^2+2\scalphi_\theta^2)d\theta ,
	\end{equation}
	and~$\lambda$-asymptotics together with the asymptotics of~$\alpha$ determined in Lemma~\ref{lemm_estimatealphaVconstant} implies that for every~$a>1/2$ there is a~$T\ge t_1$ such that
	\begin{equation}
		\frac{d\Ebas}{dt}\le \frac at \Ebas,
	\end{equation}
	for all~$t\ge T$. This concludes the proof.
\end{proof}
\begin{rema}
	For a non-constant potential, the method used in the previous proof fails, as the additional term in the derivative of the integrand of the energy~$\Ebas$, the term
	\begin{equation}
		{-}4t^{3/2}\alpha^{{-}1/2}e^{\lambda/2}\potV' \scalphi_t,
	\end{equation}
	potentially has the wrong sign. Further, it is not obvious how to bound this term in terms of the energy.
	We expect that the energy has to be adapted, and some additional a priori assumption on the potential~$\potV$ has to be made in order to gain control over the asymptotic behaviour of the individual variables.
\end{rema}

From the preliminary estimate on the energy~$\Ebas$ in Lemma~\ref{lemm_ebasestimatewitha} we can conclude a stronger one and also gain more knowledge on the asymptotic behaviour of~$\lambda$ as well as, in the following lemma, a lower bound on~$\alpha$ and first estimates on~$P$, $Q$ and~$\scalphi$.
\begin{lemm}
\label{lemm_ebasestimateoneovertwo}
	\consconstVandTtwoandlambdaasympt
	Then there is a constant~$C>0$ such that
	\begin{align}
		\normCzero{\hat\lambda(t,\cdot)}=\normCzero{\lambda(t,\cdot)+3\ln t - 2\ln(\frac{3}{4\constpot})} &{}\le C t^{{-}1/2},\\
		\Ebas(t) &{}\le C t^{1/2},
	\end{align}
	for all~$t\ge t_1$.
\end{lemm}
\begin{proof}
	Young's inequality applied to the evolution equation~\eqref{eqn_spaceevollambda} reveals
	\begin{equation}
		\absval{\lambda_\theta}\le t\alpha^{{-}1/2}\left[P_t^2+\alpha P_\theta^2+e^{2P}(Q_t+\alpha Q_\theta^2)+2(\scalphi_t^2+\alpha\scalphi_\theta^2)\right],
	\end{equation}
	and together with the estimate for the energy~$\Ebas$, Lemma~\ref{lemm_ebasestimatewitha}, this  implies
	\begin{equation}
		\int_{\SS^1}\absval{\lambda_\theta}=\O(t^{a-1}).
	\end{equation}
	As the potential is constant, we can apply the estimates for~$\alpha$ and~$\Ebas$, Lemma~\ref{lemm_estimatealphaVconstant} and Lemma~\ref{lemm_ebasestimatewitha}, to the evolution equation~\eqref{eqn_timeevollambda} and find that the mean over~$\SS^1$ of~$\lambda_t$ satisfies
	\begin{equation}
		\mean{\lambda_t}={-}4t^{1/2}\mean{e^{\lambda/2}}\constpot + \O(t^{a-5/2}).
	\end{equation}
	Here, we have used the notation introduced in Subsection~\ref{subsection_notationmean}. Recalling the definition of~$\hat\lambda$, equation~\eqref{eqn_defihatlambda}, this gives
	\begin{equation}
		\mean{\hat\lambda_t} = \frac3t(1-\mean{e^{\hat\lambda/2}})+ \O(t^{a-5/2}),
	\end{equation}
	which is the same statement as~\cite[eq.~(92)]{andreassonringstrom_cosmicnohairT3gowdyeinsteinvlasov}. As was done in the proof of Lemma~46 in the same reference, a proof by contradiction now shows that~$\mean{\hat\lambda}$ converges to zero, and in terms of a quantitative estimate, reveals that there is a constant~$C$ such that
	\begin{equation}
		\absval{\mean{\hat\lambda}}\le Ct^{{-}1/2}.
	\end{equation}
	Further, it follows that the improved estimate~$\Ebas\le Ct^{1/2}$ holds.
\end{proof}
\begin{lemm}
\label{lemm_PQphibounded}
	\consconstVandTtwoandlambdaasympt
	Then there is a constant~$C>0$ such that
	\begin{align}
		\mean{\alpha^{{-}1/2}(t,\cdot)} &{}\le C t^{3/2},\\
		\normCzero{P(t,\cdot)} + \normCzero{Q(t,\cdot)} + \normCzero{\scalphi(t,\cdot)} &{}\le C,
	\end{align}
	for all~$t\ge t_1$.
\end{lemm}
\begin{proof}
	We compute, using the evolution equation for~$\alpha$, equation~\eqref{eqn_timeevolalpha}, 
	and the estimates for~$\Ebas$ and~$\hat\lambda$ from Lemma~\ref{lemm_ebasestimateoneovertwo}, that
	\begin{align}
		\partt\mean{\alpha^{{-}1/2}} ={}& {-}\frac12\mean{\alpha^{{-}1/2}\frac{\alpha_t}{\alpha}}\\
		  ={}& \frac{1}{2\pi}\int_{\SS^1}\alpha^{{-}1/2}(\frac{e^{\lambda/2-P}J^2}{2t^{5/2}} + \frac{e^{\lambda/2+P}(K-QJ)^2}{2t^{5/2}})d\theta + \frac{1}{2\pi}\int_{\SS^1} 2t^{1/2}\alpha^{{-}1/2}e^{\lambda/2}\constpot d\theta\\
		  \le{}& Ct^{{-}1/2} +\frac{3}{2t}\mean{e^{\hat\lambda/2}\alpha^{{-}1/2}}\\
		  \le{}& \frac{3}{2t}\mean{\alpha^{{-}1/2}} + Ct^{{-}3/2}\mean{\alpha^{{-}1/2}} + Ct^{{-}1/2},
	\end{align}
	where in the last estimate, we used Taylor expansion of the estimate for~$e^{\hat\lambda}$ from the previous lemma. The rest of the proof is identical to that of~\cite[Lemma~47]{andreassonringstrom_cosmicnohairT3gowdyeinsteinvlasov}: From the previous expression, one first concludes that~$\mean{\alpha^{{-}1/2}}\le Ct^{3/2}$. This enables one to estimate
	\begin{equation}
		\int_{\SS^1}\absval{P_\theta} d\theta\le C,\qquad \absval{\partt\mean P}\le Ct^{{-}3/2},
	\end{equation}
	for all~$t\ge t_1$, using the energy estimate from Lemma~\ref{lemm_ebasestimateoneovertwo} and the upper bound on~$\alpha$ shown in Lemma~\ref{lemm_estimatealphaVconstant}.
	Combining these two estimates for the partial derivatives of~$P$ implies
	\begin{equation}
		\normCzero{P(t,\cdot)}\le C,
	\end{equation}
	for all~$t\ge t_1$. The two other variables are treated by the same method.
\end{proof}

\begin{lemm}
\label{lemm_quotientsboundstrong}
	\consconstVandTtwoandlambdaasympt
	Then there is a constant~$C>0$ such that
	\begin{equation}
		\normCzero{\frac{e^{\lambda/2-P}J^2}{t^{5/2}}} + \normCzero{\frac{e^{\lambda/2+P}(K-QJ)^2}{t^{5/2}}} \le C t^{{-}4},
	\end{equation}
	for all~$t\ge t_1$.
\end{lemm}
\begin{rema}
	Note that the decay rates we obtain for the quotients are immediately as strong as the ones in~\cite[eq.~(131)]{andreassonringstrom_cosmicnohairT3gowdyeinsteinvlasov}. We do not have  to take an intermediate step as was done in the Vlasov case in~\cite[Lemma~48]{andreassonringstrom_cosmicnohairT3gowdyeinsteinvlasov}. The main reason for this is that~$J$, $K$ are constants for non-linear scalar field solutions. 
\end{rema}
\begin{proof}
	We know that~$J$, $K$ are constant, and~$P$, $Q$ are bounded in~$\normCzero\cdot$ by the previous lemma. Further~$e^{\lambda/2}\le Ct^{{-}3/2}$ due to Lemma~\ref{lemm_ebasestimateoneovertwo}. This concludes the proof.
\end{proof}
So far, we have proven estimates for the zeroth derivatives of all variables in the evolution equations. We now proceed to first derivatives, starting with~$P$, $Q$ and~$\scalphi$.

\begin{lemm}
\label{lemm_PQphifirstderivatives}
	\consconstVandTtwoandlambdaasympt
	Then there is a constant~$C>0$ such that
	\begin{equation}
		t^3\normCzero{P_t^2+\alpha P_\theta^2 + e^{2P}(Q_t^2+\alpha Q_\theta^2)+\scalphi_t^2+\alpha\scalphi_\theta^2}\le C,
	\end{equation}
	for all~$t\ge t_1$.
\end{lemm}
\begin{rema}
	This estimate should be compared to the one obtained in Lemma~\ref{lemm_Enewbound}. The use of the energy~$\Enew$ gave us direct access at an integral estimate for~$P$, $Q$ and~$\scalphi$, but only provided boundedness by a constant. By using the energy~$\Ebas$ and proving several immediate steps, we are able to conclude a decay of order~$t^{{-}3}$ instead of only boundedness.
\end{rema}

\begin{proof}
	The proof proceeds in the spirit of the proof of~\cite[Lemma~53]{andreassonringstrom_cosmicnohairT3gowdyeinsteinvlasov}, where estimates for~$P$ and $Q$ are obtained for the Vlasov case. We give a full proof both for the variables~$P$, $Q$ as well as for~$\scalphi$, as the details differ.

	From the signs in the evolution equation for~$\alpha$, equation~\eqref{eqn_timeevolalpha}, and the estimate for~$\lambda$ found in Lemma~\ref{lemm_ebasestimateoneovertwo}, we conclude that
	\begin{equation}
	\label{eqn_estimtimederivalpha}
		{-}(\frac2t-\frac{\alpha_t}{\alpha})\le {-}(\frac2t + 4t^{1/2}e^{\lambda/2}\constpot) = {-}\frac5t +\O(t^{{-}3/2}).
	\end{equation}
	
	Let us start with the estimates on~$P$, $Q$.
	A short computation of the individual terms of the function~$\Afctmp$ defined in equation~\eqref{eqn_defiAfctpm} shows that
	\begin{align}
		\mp\frac2t\alpha^{1/2}(P_\mp P_\theta +e^{2P}Q_\mp Q_\theta)
		  ={}&\frac{1}{2t}(\Afctmp-\Afctpm)+\frac2t\alpha(P_\theta^2+e^{2P} Q_\theta^2)\\
		  \le{}&\frac1t \Afctmp +\frac1{2t}(\Afct_++\Afct_-).
	\end{align}
	With this, the lightcone evolution of~$\Afctmp$ which we have given in Lemma~\ref{lemm_Afctlightconederiv} can be estimated via
	\begin{equation}
		\partpm\Afctmp\le {-}\frac4t\Afctmp + Ct^{{-}3/2}\Afctmp + \frac{1}{2t}(\Afct_++\Afct_-) +Ct^{{-}5}\Afctmp^{1/2},
	\end{equation}
	where we have used estimate~\eqref{eqn_estimtimederivalpha} and Lemma~\ref{lemm_quotientsboundstrong}.
	We set
	\begin{equation}
		\hatAfctmp=t^4\Afctmp+t
	\end{equation}
	and find
	\begin{equation}
		\partpm\hatAfctmp\le\frac{1}{2t}(\hat\Afct_++\hat\Afct_-) + Ct^{{-}3/2}\hatAfctmp.
	\end{equation}
	Defining
	\begin{equation}
		\hat F(t)=\sup_{\theta\in\SS^1}\hat\Afct_+(t,\theta)+\sup_{\theta\in\SS^1}\hat\Afct_-(t,\theta),
	\end{equation}
	we can use this estimate on~$\partpm\hatAfctmp$ together with Grönwall's lemma to obtain
	\begin{equation}
		\hat F(t)\le Ct.
	\end{equation}
	Consequently,
	\begin{equation}
		P_t^2+\alpha P_\theta^2 + e^{2P}(Q_t^2+\alpha Q_\theta^2) = \frac12(\Afct_++\Afct_-)\le Ct^{{-}3}.
	\end{equation}

	We now turn our attention to~$\scalphi$. A similar computation as for~$\Afctmp$ shows that
	\begin{equation}
		\mp\frac2t\alpha^{1/2}\scalphi_\mp \scalphi_\theta
		  =\frac{1}{2t}(\Bfctmp-\Bfctpm)+\frac2t\alpha \scalphi_\theta^2
		  \le\frac1t \Bfctmp +\frac1{2t}(\Bfct_++\Bfct_-),
	\end{equation}
	and applying this and estimate~\eqref{eqn_estimtimederivalpha} to Lemma~\ref{lemm_Bfctlightconederiv} yields
	\begin{equation}
		\partpm\Bfctmp\le {-}\frac4t\Bfctmp + Ct^{{-}3/2}\Bfctmp + \frac{1}{2t}(\Bfct_++\Bfct_-).
	\end{equation}
	This is a (simplified) version of the estimate we had available for~$\partpm\Afctmp$. The whole chain of arguments leading to the estimates for~$P$, $Q$ builds on this form, which means that we can use the same steps to find
	\begin{equation}
		\scalphi_t^2+\alpha \scalphi_\theta^2  = \frac12(\Bfct_++\Bfct_-)\le Ct^{{-}3}.
	\end{equation}
\end{proof}

Having obtained estimates for the first derivatives of~$P$, $Q$ and~$\scalphi$, we can also tackle the first derivatives of the remaining variables, and improve our grasp on the zeroth derivatives.
\begin{lemm}
\label{lemm_alphalambdaderivimproved}
	\consconstVandTtwoandlambdaasympt
	Then there is a constant~$C>0$ such that 
	\begin{align}
		\normCzero{\hat\lambda(t,\cdot)}=\normCzero{\lambda(t,\cdot)+3\ln t-2\ln(\frac{3}{4\constpot})}\le {}& Ct^{{-}1},\\
		\normCzero{\frac{\alpha_t}{\alpha}+\frac3t} + \normCzero{\lambda_t+\frac3t}\le {}& Ct^{{-}2},\\
		\normCzero{\alpha^{1/2}\lambda_\theta} \le{}& Ct^{{-}2},\\
		\normCzero{\parttheta(\frac{e^{\lambda/2-P}J^2}{t^{5/2}})} + \normCzero{\parttheta(\frac{e^{\lambda/2+P}(K-QJ)^2}{t^{5/2}})} \le {}& C t^{{-}4},
	\end{align}
	for all~$t\ge t_1$.
\end{lemm}
\begin{proof}[Proof of Lemma~\ref{lemm_alphalambdaderivimproved}]
	Returning to the evolution for~$\partt\hat\lambda$, equation~\eqref{eqn_evolhatlambda}, and applying the results of Lemma~\ref{lemm_PQphifirstderivatives}, we find that
	\begin{equation}
		\partt\hat\lambda=\frac3t-\frac3te^{\hat\lambda/2}+\O(t^{{-}2})
	\end{equation}
	and conclude
	\begin{equation}
		\partt\hat\lambda^2\le {-}\frac3t \hat\lambda^2+Ct^{{-}2}\absval{\hat\lambda}
	\end{equation} via Taylor expansion and the estimate for~$\hat\lambda$ in Lemma~\ref{lemm_ebasestimateoneovertwo}.
	From this, we obtain the first estimate in the statement by showing that~$t^2\hat\lambda^2$ is bounded, as in the proof of~\cite[Lemma~54]{andreassonringstrom_cosmicnohairT3gowdyeinsteinvlasov}. 
	
	For the estimates on~$\alpha_t/\alpha$, $\lambda_t$ and~$\lambda_\theta$ in the statement, we consider equations~\eqref{eqn_timeevolalpha}, \eqref{eqn_timeevollambda} and~\eqref{eqn_spaceevollambda} and apply the estimates from Lemma~\ref{lemm_quotientsboundstrong}, Lemma~\ref{lemm_PQphifirstderivatives} as well as the estimate for~$\lambda$ we have just obtained.
	
	As a direct consequence of this, we find that~$t^3\alpha$ converges to a strictly positive function, in other words
	\begin{equation}
	\label{eqn_improvedestimalpha}
		C_1t^{{-}3} \le \alpha(t,\theta) \le C_2t^{{-}3}, 
	\end{equation}
	for suitable constants~$C_{1,2}>0$, $t\ge t_1$. 
	
	Lastly, the estimates in the statement for the~$\theta$-derivatives of quotients are a direct consequence of the estimates for the non-differentiated quotients from Lemma~\ref{lemm_quotientsboundstrong} together with Lemma~\ref{lemm_PQphifirstderivatives} for~$P_\theta$, the estimate for~$\lambda_\theta$ and the lower bound on~$\alpha$ we just showed, and the fact that~$P$, $Q$ are bounded due to Lemma~\ref{lemm_PQphibounded}.
\end{proof}

\begin{lemm}
\label{lemm_alphaderivtheta}
	\consconstVandTtwoandlambdaasympt
	Then there is a constant~$C>0$ such that
	\begin{align}
		\normCzero{\parttheta(\frac{\alpha_t}{\alpha})}\le {}& Ct^{{-}3/2},\\
		\normCzero{\frac{\alpha_\theta}{\alpha}} \le {}& C,
	\end{align}
	for all~$t\ge t_1$.
\end{lemm}
\begin{proof}
	Differentiating the evolution equation~\eqref{eqn_timeevolalpha} with respect to~$\theta$, we find
	\begin{equation}
		\parttheta(\frac{\alpha_t}{\alpha})={-}\parttheta(\frac{e^{\lambda/2-P}J^2}{t^{5/2}} + \frac{e^{\lambda/2+P}(K-QJ)^2}{t^{5/2}}) -\frac{\lambda_\theta}{2}4t^{1/2}e^{\lambda/2}\constpot.
	\end{equation}
	From Lemma~\ref{lemm_alphalambdaderivimproved}, we conclude that the two terms in the bracket decay as~$\O(t^{{-}4})$. Using the same Lemma, we find that the last term decays is of order~$\O(t^{{-}3/2})$. 
	This implies
	\begin{equation}
		\normCzero{\partt(\frac{\alpha_\theta}{\alpha})}=\normCzero{\parttheta(\frac{\alpha_t}{\alpha})}\le Ct^{{-}3/2},
	\end{equation}
	and integration with respect to~$t$ concludes the proof.
\end{proof}

The next  and final part of discussing the asymptotic behaviour of the individual functions in the system of evolution equations is a longer inductive argument. We make the following inductive assumption and, using several auxiliary steps detailed in Lemma~\ref{lemm_alphalambdahigherderiv}, prove that the statement of this assumption holds for all~$N\in\NN$.
\begin{inducassum}
\label{inductiveassumpt}
	For some~$1\le {N}\in\ZZ$ there is a constant~$C_{N-1}$
	such that
	\begin{equation}
	\label{eqn_inductiveassumpt}
		\normck{N-1}{P_\theta}+\normck{N-1}{Q_\theta}+\normck{N-1}{\scalphi_\theta}
		+t^{3/2}\normck{N-1}{P_t}+t^{3/2}\normck{N-1}{Q_t}+t^{3/2}\normck{N-1}{\scalphi_t}\le C_{N-1},
	\end{equation}
	for all~$(t,\theta)\in[t_1,\infty)\times\TT^3$.
\end{inducassum}

\begin{lemm}
\label{lemm_alphalambdahigherderiv}
	\consconstVandTtwoandlambdaasympt
	Assume that the Inductive assumption~\ref{inductiveassumpt} holds for some~$1\le {N}\in\ZZ$. Then there are constants~$C_j$, $j=0,\ldots N$, depending only on the solution and on~$N$, such that
	\begin{align}
		\normCzero{\parttheta^{j}\lambda} \le{}& C_j t^{{-}1/2},\\
		\normCzero{\alpha^{{-}1}\parttheta^j\alpha} \le{}& C_j,\\
		\normCzero{\parttheta^{j}(\frac{\alpha_t}{\alpha})} \le{}& C_j t^{{-}3/2},
	\end{align}
	for~$t\ge t_1$ and~$1\le j\le N$.
\end{lemm}
\begin{proof}
	The case~$N=1$ is covered by Lemma~\ref{lemm_alphaderivtheta} yielding the two estimates with~$\alpha$ and the estimate for~$\lambda_\theta$ from Lemma~\ref{lemm_alphalambdaderivimproved}. We can therefore assume that~$N\ge2$.
	
	For each~$1\le j\le N$, differentiating the equation for~$\lambda_\theta$, equation~\eqref{eqn_spaceevollambda}, $j-1$ times with respect to~$\theta$ and employing the Inductive assumption~\ref{inductiveassumpt} yields
	\begin{equation}
	\label{eqn_auxilhigherorderlambda}
		\normCzero{\parttheta^{j}\lambda} \le C_j t^{{-}1/2}.
	\end{equation}
	Further, the evolution equation~\eqref{eqn_timeevolalpha} reveals
	\begin{equation}
		\parttheta^{j}(\frac{\alpha_t}{\alpha}) = \parttheta^{j}({-}\frac{e^{\lambda/2-P}J^2}{t^{5/2}} - \frac{e^{\lambda/2+P}(K-QJ)^2}{t^{5/2}} - 4t^{1/2}e^{\lambda/2}\constpot).
	\end{equation}
	The last term decays as~$t^{{-}3/2}$, due to estimate~\eqref{eqn_auxilhigherorderlambda} together with Lemma~\ref{lemm_alphalambdaderivimproved} giving~$e^{\lambda/2}=\O(t^{{-}3/2})$.
	In comparison, the two quotients can be neglected, as the~$\theta$-derivatives of~$P$, $Q$ as well as~$\lambda$ are at least bounded, due to the Inductive assumption~\ref{inductiveassumpt} and again estimate~\eqref{eqn_auxilhigherorderlambda}, and the quotients themselves decay of as~$t^{{-}4}$, see Lemma~\ref{lemm_quotientsboundstrong}.
	Consequently,
	\begin{equation}
		\normCzero{\partt \parttheta^{j-1}(\frac{\alpha_\theta}{\alpha})}=\normCzero{\parttheta^{j}(\frac{\alpha_t}{\alpha})}\le Ct^{{-}3/2},
	\end{equation}
	and integration yields 
	\begin{equation}
		\normCzero{\parttheta^{j-1}(\frac{\alpha_\theta}{\alpha})}\le C_j.
	\end{equation}
	As
	\begin{equation}
		\alpha^{{-}1}(\parttheta^{j}\alpha)=\parttheta^{j-1}(\frac{\alpha_\theta}{\alpha})-\sum_{1\le i\le j-1} \gamma_i\parttheta^i(\alpha^{{-}1})  \parttheta^{j-i}(\alpha),
	\end{equation}
	where~$\gamma_i$ are binomial factors,
	we can inductively conclude that the remaining estimate in the statement is true for all~$1\le j\le N$.
\end{proof}

\begin{lemm}
\label{lemm_inductassumfornextN}
	\consconstVandTtwoandlambdaasympt
	Assume that the Inductive assumption~\eqref{inductiveassumpt} holds for some~$1\le {N}\in\ZZ$. Then there is a constant~$C_{N}$, depending only on the solution and on~$N$, such that
	\begin{align}
	\label{eqn_inductionclose}
		t^{3/2}\normCzero{\parttheta^NP_t}+\normCzero{\parttheta^NP_\theta} +
		t^{3/2}\normCzero{e^P\parttheta^NQ_t}+\normCzero{e^P\parttheta^NQ_\theta} &\\
		+
		t^{3/2}\normCzero{\parttheta^N\scalphi_t}+\normCzero{\parttheta^N\scalphi_\theta} &\le C_N,
	\end{align}
	for~$t\ge t_1$. As a consequence, \eqref{eqn_inductiveassumpt} holds with~$N$ replaced by~$N+1$.
\end{lemm}
\begin{proof}
	The proof for~$P$, $Q$ is identical to that of~\cite[Lemma~64]{andreassonringstrom_cosmicnohairT3gowdyeinsteinvlasov}: One expands
	\begin{equation}
		\partpm\left[\parttheta^NP_t\mp\parttheta^N(\alpha^{1/2}P_\theta)\right], \qquad
		\partpm\left[\parttheta^NQ_t\mp\parttheta^N(\alpha^{1/2}Q_\theta)\right],
	\end{equation}
	separating the~$N$-th order terms from the lower order ones. The lower order terms are estimated using the Inductive assumption~\ref{inductiveassumpt} which in particular implies Lemma~\ref{lemm_alphalambdahigherderiv}.
	For the highest order term, one then uses the lightcone wave equation expression from Lemma~\ref{lemm_Afctlightconederiv}. 
	For the reader's convenience, we list the prerequisites which the proof in~\cite{andreassonringstrom_cosmicnohairT3gowdyeinsteinvlasov} makes use of and where these statements have been proven in our case:
	\begin{itemize}
		\item our Inductive assumption~\ref{inductiveassumpt} replaces their Inductive assumption~61,
		\item our Lemma~\ref{lemm_alphalambdahigherderiv} replaces their Lemma~63,
		\item our equation~\eqref{eqn_improvedestimalpha} replaces their equation~(135),
		\item our evolution equations~\eqref{eqn_secondevolP} and~\eqref{eqn_secondevolQ} replace their evolution equations~(61) and~(68),
		\item the combined statemens in our Lemma~\ref{lemm_quotientsboundstrong} and Lemma~\ref{lemm_alphalambdaderivimproved} replace their Lemma~54, of which their equation~(128) is a part. 
	\end{itemize}
	
	Using a similar approach, we now prove the estimates for~$\scalphi$. We compute explicitly that
	\begin{align}
		\partpm\left[ \parttheta^N\scalphi_t \mp \parttheta^N(\alpha^{1/2}\scalphi_\theta )\right]
		  ={}& \parttheta^N\scalphi_{tt} \pm\alpha^{1/2}\parttheta^{N+1}\scalphi_t \mp\parttheta^N\left[\frac{\alpha_t}{2\alpha}\alpha^{1/2}\scalphi_\theta+\alpha^{1/2}\scalphi_{t\theta}\right]\\
		      &{}-\alpha^{1/2}\parttheta^{N+1}(\alpha^{1/2}\scalphi_\theta)\\
		  ={}& \parttheta^N\scalphi_{tt} \mp\frac{\alpha_t}{2\alpha}\parttheta^N(\alpha^{1/2}\scalphi_\theta)\mp\frac{N\alpha_\theta}{2\alpha}\alpha^{1/2}\parttheta^N\scalphi_t\\
		      &{}-\parttheta^N\left[\alpha^{1/2}\parttheta(\alpha^{1/2}\scalphi_\theta)\right] +\frac{N\alpha_\theta}{2\alpha}\alpha^{1/2}\parttheta^N(\alpha^{1/2}\scalphi_\theta)+\O(t^{{-}3}),
	\end{align}
	where we have estimated all terms with lower derivatives of~$\scalphi$ and~$\alpha$ via the Inductive assumption~\ref{inductiveassumpt} and Lemma~\ref{lemm_alphalambdahigherderiv}.
	Similarly, we treat the derivative of the equation of motion~\eqref{eqn_eqnofmotionscalfieldwithT2} and find
	\begin{align}
		\parttheta^N\left[ \scalphi_{tt}-\alpha^{1/2}\parttheta(\alpha^{1/2}\scalphi_\theta)\right]
			={}& \parttheta^N\left[ \scalphi_{tt}-\alpha\scalphi_{\theta\theta}-\frac{\alpha_\theta}{2}\scalphi_\theta\right]\\
			={}& \parttheta^N\left[-\frac1t\scalphi_t+\frac{\alpha_t}{2\alpha}\scalphi_t\right]\\
			={}& {-}\frac1t\parttheta^N\scalphi_t+\frac{\alpha_t}{2\alpha}\parttheta^N\scalphi_t+\O(t^{{-}3}).
	\end{align}
	We can now introduce
	\begin{equation}
		\DfctNpm =\left[\parttheta^N\scalphi_t \pm \parttheta^N(\alpha^{1/2}\scalphi_\theta )\right]^2
	\end{equation}
	and combine the previous two computations, making use of the fact that
	\begin{align}
		\mp\frac2t \parttheta^N(\alpha^{1/2}\scalphi_\theta) (\parttheta^N\scalphi_t\mp \parttheta^N(\alpha^{1/2}\scalphi_\theta))
		={}& \frac1{2t}(\DfctNmp -\DfctNpm ) +\frac2t\left[ \parttheta^N(\alpha^{1/2}\scalphi_\theta)\right]^2 \\
		\le{}& \frac1{2t}(\DfctNmp -\DfctNpm ) + \frac1t(\DfctNp+\DfctNm),
	\end{align}
	to conclude
	\begin{align}
		\partpm\DfctNmp
		\le{}&\frac{\alpha_t}{\alpha}\DfctNmp \mp\frac{N\alpha_\theta}{\alpha}\alpha^{1/2}\DfctNmp {-}\frac2t\parttheta^N\scalphi_t (\parttheta^N\scalphi_t\mp \parttheta^N(\alpha^{1/2}\scalphi_\theta))\\
		&{}+ C_N t^{{-}3}(\DfctNmp)^{1/2}\\
		\le{}&{-}\frac5t\DfctNmp \mp\frac2t \parttheta^N(\alpha^{1/2}\scalphi_\theta) (\parttheta^N\scalphi_t\mp \parttheta^N(\alpha^{1/2}\scalphi_\theta)) +C_N t^{{-}3}(\DfctNmp)^{1/2}\\
		&{}+ C_Nt^{{-}3/2}(\DfctNm+\DfctNm)\\
		\le{}&{-}\frac5t\DfctNmp +\frac1{2t}(\DfctNmp -\DfctNpm ) + \frac1t(\DfctNp+\DfctNm)\\
		&{}+C_N t^{{-}3}(\DfctNmp)^{1/2}+ C_Nt^{{-}3/2}(\DfctNp+\DfctNm).
	\end{align}
	Here, we have made use of the asymptotic properties of~$\alpha_t$ and~$\alpha_\theta$ from Lemma~\ref{lemm_alphalambdaderivimproved} and Lemma~\ref{lemm_alphaderivtheta}.
	Setting
	\begin{equation}
		\hatDfctNpm=t^{7/2}\DfctNpm+t^{1/2},\quad \hat G_{N+1,\pm}=\sup_{\theta\in\SS^1}\hatDfctNpm
	\end{equation}
	and~$\hat G_{N+1}=\hat G_{N+1,\pm}+\hat G_{N+1,\mp}$,
	this implies
	\begin{equation}
		\partpm\hatDfctNmp\le\frac{1}{2t}\hatDfctNpm+ C_Nt^{{-}3/2}(\hatDfctNp+\hatDfctNm),
	\end{equation}
	from which we can even conclude
	\begin{equation}
		\hat G_{N+1}(t)\le  \hat G_{N+1}(t_1)+\int_{t_1}^t(\frac{1}{2s}\hat G_{N+1}(s) + C_Ns^{{-}3/2} \hat G_{N+1}(s))ds.
	\end{equation}
	Grönwall's lemma therefore yields that~$\hat G_{N+1}(t)\le C_Nt^{1/2}$. Together with the Inductive assumption~\ref{inductiveassumpt}, Lemma~\ref{lemm_alphalambdahigherderiv} and estimate~\eqref{eqn_improvedestimalpha}, this concludes the proof.
\end{proof}
\begin{coro}
\label{coro_collectasymestim}
	\consconstVandTtwoandlambdaasympt
	Then there is a constant~$C_N$, depending only on~$N$ and the solution, such that
	\begin{align}
		\normck N{P_t}+ \normck N{Q_t} + \normck N{\scalphi_t} + \normck N{\frac{\alpha_t}{\alpha}+\frac3t} + \normck N{\lambda_t+\frac3t} +t^{1/2}\normck N{\alpha^{1/2}\lambda_\theta}\le C_N t^{{-}2},
	\end{align}
	for all~$t\ge t_1$.
\end{coro}
\begin{proof}
	By Lemma~\ref{lemm_inductassumfornextN}, we know that the Inductive assumption~\ref{inductiveassumpt} holds for all~$1\le N\in\ZZ$. As a consequence, all conclusions of Lemma~\ref{lemm_alphalambdahigherderiv} hold for all~$1\le N\in\ZZ$. We combine this with the second order evolution equations for~$P$, $Q$ and~$\scalphi$, equations~\eqref{eqn_secondevolP}, \eqref{eqn_secondevolQ} and~\eqref{eqn_secondevolscalphifromeqnmotion}, to find
	\begin{equation}
		\normck N{\partt(t\alpha^{{-}1/2}\scalphi_t)}
		+\normck N{\partt(t\alpha^{{-}1/2}P_t)}
		+\normck N{\partt(t\alpha^{{-}1/2}Q_t)}
		\le C_Nt^{{-}1/2}.
	\end{equation}
	Integrating these expressions and making use of estimates~\eqref{eqn_improvedestimalpha} for~$\alpha$  and Lemma~\ref{lemm_alphalambdahigherderiv}, we conclude that
	\begin{equation}
		\normck N{\scalphi_t}+ \normck N{P_t}+ \normck N{Q_t}\le C_N t^{{-}2}.
	\end{equation}
	For the remaining estimates, we combine this last estimate with the evolution equations~\eqref{eqn_timeevolalpha}, \eqref{eqn_timeevollambda} and~\eqref{eqn_spaceevollambda}, again the estimates for~$\alpha$, and Lemma~\ref{lemm_alphalambdaderivimproved} to conclude.
\end{proof}

\section{Proof of the main statements on the \nohairconj}
\label{section_mainproofs}

We now provide proofs of the main statements from the introduction which discuss the \nohairconj. For this, we heavily rely on the work carried out in the previous section.
\begin{proof}[Proof of Proposition~\ref{prop_mainasympmetric}]
	The estimates on~$P$, $Q$ and~$\scalphi$, both for their zeroth and first derivative, are a direct consequence of Corollary~\ref{coro_collectasymestim}. The same is true for the estimates of the first derivatives of~$\alpha$ and~$\lambda$. For their zeroth derivatives, we combine these two last estimates with Lemma~\ref{lemm_ebasestimateoneovertwo} and estimate~\eqref{eqn_improvedestimalpha}. 
	
	We turn to~$G$ and~$H$, which satisfy
	\begin{equation}
		H_t={-}t^{{-}5/2}\alpha^{{-}1/2}e^{\lambda/2+P}(K-QJ),\quad 
		G_t={-}QH_t-t^{{-}5/2}\alpha^{{-}1/2}e^{\lambda/2-P}J,
	\end{equation}
	according to~\eqref{eqn_twistanalytic}. We recall that~$J$, $K$ are constants, and that~$P$, $Q$ are bounded in every~$C^N$-norm by Corollary~\ref{coro_collectasymestim}. Combining this with  Lemma~\ref{lemm_ebasestimateoneovertwo} and estimate~\eqref{eqn_improvedestimalpha} for the zeroth derivative of~$\alpha$ and~$\lambda$ as well as Lemma~\ref{lemm_alphalambdahigherderiv} for the higher ones, this yields the estimates of~$G$ and~$H$.

	All which remains to show is the geometric estimates. The proof can be carried over from that of Theorem~7 in~\cite{andreassonringstrom_cosmicnohairT3gowdyeinsteinvlasov}, page~52.
\end{proof}
\begin{proof}[Proof of Theorem~\ref{theo_nohair}]
	We do not carry out the details of the proof here but refer the reader to the very detailed arguments in~\cite{andreassonringstrom_cosmicnohairT3gowdyeinsteinvlasov}, where the proof of their Theorem~14 is given on page~53. As is the case there, the \nohairconj\ follows from the estimates on the metric and \fundform\ which, for the case of a non-linear scalar field with~$\TT^2$-symmetry and~$\lambda$-asymptotic, we have provided in Proposition~\ref{prop_mainasympmetric}.
	The steps in the proof are as follows: 
	One starts by noting that the sets~$\{t\}\times\TT^3$ are Cauchy hypersurfaces. Using the limits and convergence rates of the individual variables appearing in the metric and the \fundform, one finds bounds on the injectivity radius of the late time boundary hypersurface~$(\TT^3,\bar g_\infty)$. These bounds ensure existence of a small neighborhood of a fixed point~$\bar x\in\TT^3$ where normal coordinates can be defined. Via the normal coordinates, it is possible to find a four-dimensional set large enough to contain an open, non-empty subset~$D$ diffeomorphic to~$\setHKT$. Showing that the metric and \fundform\ induced by this diffeomorphism have the correct asymptotic behaviour is a matter of applying the known estimates on~$\bar g$ and~$\bar k$.
\end{proof}

\begin{proof}[Proof of Proposition~\ref{prop_geodcomplete}]
	The proof proceeds similarly to that of~\cite[Prop.~4]{ringstrom_futstabnonlinscalfield}. Consider a future directed causal geodesic~$c$ and assume that its maximal interval of existence is~$(s_{\min},s_{\max})$. To conclude, we have to show that~$s_{\max}=\infty$.
	
	We use the notation~$t=c^0(s)$ and introduce the \onorm\ frame~$e_0$, $e_1$, $e_2$, $e_3$ as in~\eqref{eqns_onormframe}. The geodesic equation for~$c$ implies that
	\begin{equation}
	\label{eqn_geodesic}
		\ddot c^0 +\Gamma^0_{\mu\nu}\dot c^\mu \dot c^\nu=0.
	\end{equation}
	The Christoffel symbols~$\Gamma^0_{\mu\nu}$ with respect to the chosen frame have been explicitly computed for a~$\TT^2$-symmetric metric of the form~\eqref{eqn_metricT2} in~\cite[App.~A]{andreassonringstrom_cosmicnohairT3gowdyeinsteinvlasov}, and the expressions are given in terms of the variables~$\alpha>0$, $\lambda$, $P$, $Q$, $J$ and~$K$. Combining these expressions with the asymptotic behaviour we proved in Proposition~\ref{prop_mainasympmetric}, we find that
	\begin{equation}
		\absval{\Gamma_{00}^0}=0,\qquad
		\absval{\Gamma_{0i}^0}\le Ct^{{-}3/2},\qquad
		\absval{\Gamma_{ij}^0-\H \delta_{ij}}\le Ct^{{-}1},
	\end{equation}
	where we have set~$\H = (\constpot/3)^{1/2}$.
	As a consequence, $\Gamma_{ij}^0\dot c^i\dot c^j\ge 0$ for sufficiently large times~$t$.
	Due to causality and future directedness of the curve together with \onorm ity of the frame, we find that
	\begin{equation}
	\label{eqn_causalgeod}
		{-}(\dot c^0)^2 + \delta_{ij}\dot c^i \dot c^j \le 0,\qquad
		{-}\dot c^0 <0.
	\end{equation}
	With this, equation~\eqref{eqn_geodesic} implies
	\begin{equation}
		\ddot c^0\le Ct^{{-}3/2}\dot c^0 \dot c^0.
	\end{equation}
	As in the proof of~\cite[Prop.~4]{ringstrom_futstabnonlinscalfield}, we conclude from this estimate that~$\dot c^0$ is bounded for~$s\ge s_1$, as well as
	\begin{equation}
		c^0(s)-c^0(s_1)\le C\absval{s-s_1}.
	\end{equation}
	
	Suppose that~$s_{\max}$ is finite. This implies that~$c^0(s)$ remains bounded for all~$s\in(s_1,s_{\max})$, and due to~$\dot c^0>0$ has to converge to a finite number. The first estimate in~\eqref{eqn_causalgeod} implies that the same holds true for~$c^i$, $i=1,2,3$.
	Consequently, $c$ converges to a point in the spacetime as~$s\rightarrow s_{\max}$ and can therefore be extended beyond~$s_{\max}$, a contradiction.
\end{proof}

\appendix

\section{Finding the evolution equations for scalar field}
\label{appendix_evolequsscalarfield}

The aim of this section is to explain the provenance of the evolution equations~\eqref{eqn_timeevolalpha}--\eqref{eqn_eqnofmotionscalfieldwithT2}. For this, we partially rely on results obtained in~\cite{andreassonringstrom_cosmicnohairT3gowdyeinsteinvlasov}.

We investigate the~$\TT^2$-symmetric metric given in equation~\eqref{eqn_metricT2} and, as in~\cite[eq~(7)]{andreassonringstrom_cosmicnohairT3gowdyeinsteinvlasov}, introduce the \onorm\ frame
\begin{equation}
\label{eqns_onormframe}
	\begin{subaligned}
		e_0={}&t^{1/4}e^{{-}\lambda/4}\partial_t, &&
		e_1={}t^{1/4}e^{{-}\lambda/4}\alpha^{1/2}(\partial_\theta-G\partial_x-H\partial_y),\\
		e_2={}&t^{{-}1/2}e^{{-}P/2}\partial_x, &&
		e_3={}t^{{-}1/2}e^{P/2}(\partial_y-Q\partial_x).
	\end{subaligned}
\end{equation}

In order to connect to the results in~\cite{andreassonringstrom_cosmicnohairT3gowdyeinsteinvlasov}, we now need to compute the different components of the stress-energy tensor of a non-linear scalar with respect to this frame. To this end, we first notice the following invariance property.
\begin{lemm}
\label{lemm_T2symmderivativesphivanish}
	Consider a~$\TT^2$-symmetric metric~\eqref{eqn_metricT2} satisfying the Einstein non-linear scalar field equations~\eqref{eqn_stressenergytensorgeneralscalarfield} with scalar field~$\scalphi$, potential~$\potV$ and equation of motion~\eqref{eqn_eqnofmotiongeneralscalarfield}. 
	Then~$\scalphi_x=0=\scalphi_y$ as well as~$\nabla_{e_2}\scalphi=0=\nabla_{e_3}\scalphi$.
\end{lemm}
\begin{proof}
	As the spacetime with non-linear scalar field~$(I\times\TT^3,g,\scalphi)$ is~$\TT^2$-symmetric by assumption, the metric~$g$ and the scalar field~$\scalphi$ are invariant under the isometries described after equation~\eqref{eqn_metricT2}. This implies~$\scalphi_x=0=\scalphi_y$. The remaining two equations in the statement follow by definition of the vector fields~$e_2$, $e_3$.
\end{proof} 

Given an \onorm\ frame~$e_0$, $e_1$, $e_2$, $e_3$, the different components of the stress-energy tensor are denoted as follows:
\begin{equation}
	\rho\coloneqq T(e_0,e_0), \quad 
	P_i\coloneqq(e_i,e_i), \quad
	J_i\coloneqq{-}T(e_0,e_i)\quad
	S_{ij}=T(e_i,e_j),	
\end{equation}
where~$i\not=j$ and we do not sum over indices occuring twice.

Applying the previous lemma, a straight-forward computation now yields that in a~$\TT^2$-symmetric non-linear scalar field these quantities are given by
\begin{align}
	\rho={}&\frac12(\nabla_0\scalphi)^2+\frac12(\nabla_1\scalphi)^2+V(\scalphi)\\
	    ={}&\frac12t^{1/2}e^{{-}\lambda/2}(\scalphi_t^2+\alpha\scalphi_\theta^2)+V(\scalphi),\\
	P_1={}&\frac12(\nabla_0\scalphi)^2+\frac12(\nabla_1\scalphi)^2-V(\scalphi)\\
	    ={}&\frac12t^{1/2}e^{{-}\lambda/2}(\scalphi_t^2+\alpha\scalphi_\theta^2)-V(\scalphi),\\
	P_2={}&\frac12(\nabla_0\scalphi)^2-\frac12(\nabla_1\scalphi)^2-V(\scalphi)\\
	    ={}&\frac12t^{1/2}e^{{-}\lambda/2}(\scalphi_t^2-\alpha\scalphi_\theta^2)-V(\scalphi)=P_3.
\end{align}
Further, 
\begin{align}
	J_1={}&{-}\nabla_0\scalphi\nabla_1\scalphi={-}t^{1/2}e^{{-}\lambda/2}\alpha^{1/2}\scalphi_t\scalphi_\theta,\\
	J_2={}&0=J_3,
\end{align}
and
\begin{equation}
	S_{ij}=T(e_i,e_j)={}0,\qquad i\not=j.
\end{equation}
Replacing the occurences of~$\rho$, $P_i$, $J_i$, and~$S_{ij}$ in the evolution equations~(60)--(69) in~\cite{andreassonringstrom_cosmicnohairT3gowdyeinsteinvlasov} by these expressions and omitting the cosmological constant, we obtain our evolution equations~\eqref{eqn_timeevolalpha}--\eqref{eqn_secondevollambda} as well as
\begin{equation}
	J_\theta=K_\theta=J_t=K_t=0.
\end{equation}

The equation of motion in coordinates, equation~\eqref{eqn_eqnofmotionscalfieldwithT2}, is an immediate consequence of the general equation of motion of a non-linear scalar field, equation~\eqref{eqn_eqnofmotiongeneralscalarfield}, applying Lemma~\ref{lemm_T2symmderivativesphivanish} and using the definition of the frame elements~$e_0,e_1$.

\bibliographystyle{amsalpha}
\bibliography{researchbib}
\vfill

\end{document}